\documentclass[11pt]{article}
\usepackage{amsfonts,amssymb}
\usepackage{titlesec}
\usepackage{titletoc}
\usepackage{amsmath}
\usepackage{listings}
\usepackage{verbatim}
\usepackage{graphicx}
\usepackage{geometry}
\usepackage[colorlinks=true,linkcolor=blue,citecolor=red,urlcolor=green]{hyperref}
\usepackage{mathrsfs}
\usepackage{setspace}
\usepackage{cite}
\usepackage{bbm}
\usepackage{mathtools}
\usepackage{geometry} 
\usepackage{listings}
\usepackage{amssymb,amsmath,amsfonts,latexsym}
\usepackage{amsmath,graphicx,bm,xcolor,url}
\usepackage[caption=false]{subfig} 
\usepackage{fixltx2e}%ordering of single and double column floats
\usepackage{array}%array and tabular environments
\usepackage{verbatim}
\usepackage{bm}
\usepackage{verbatim}
\usepackage{textcomp}
\usepackage{mathrsfs}
\usepackage{epstopdf}
\usepackage{relsize}
\usepackage{cleveref} 
\usepackage{subfig}
 \usepackage{amsthm}

\catcode`~=11 \def\UrlSpecials{\do\~{\kern -.15em\lower .7ex\hbox{~}\kern .04em}} \catcode`~=13 

\allowdisplaybreaks[3]

\newcommand{\nn}{\nonumber}

\newcommand{\calA}{\mathcal{A}}

\newcommand{\calC}{\mathcal{C}}

\newcommand{\calK}{\mathcal{K}}
\newcommand{\calL}{\mathcal{L}}

\newcommand{\calN}{\mathcal{N}}

\newcommand{\calS}{\mathcal{S}}
\newcommand{\calT}{\mathcal{T}}

\newcommand{\calW}{\mathcal{W}}

\newcommand{\ba}{\mathbf{a}}
\newcommand{\bA}{\mathbf{A}}

\newcommand{\bB}{\mathbf{B}}

\newcommand{\be}{\mathbf{e}}

\newcommand{\bg}{\mathbf{g}}
\newcommand{\bG}{\mathbf{G}}
\newcommand{\bh}{\mathbf{h}}

\newcommand{\bi}{\mathbf{i}}
\newcommand{\bI}{\mathbf{I}}

\newcommand{\bP}{\mathbf{P}}

\newcommand{\bQ}{\mathbf{Q}}

\newcommand{\bs}{\mathbf{s}}

\newcommand{\bu}{\mathbf{u}}
\newcommand{\bU}{\mathbf{U}}
\newcommand{\bv}{\mathbf{v}}
\newcommand{\bV}{\mathbf{V}}
\newcommand{\bw}{\mathbf{w}}
\newcommand{\bW}{\mathbf{W}}
\newcommand{\bx}{\mathbf{x}}
\newcommand{\bX}{\mathbf{X}}
\newcommand{\by}{\mathbf{y}}

\newcommand{\bz}{\mathbf{z}}

\newcommand{\bphi}{\bm{\phi}}

\newcommand{\bSigma	}{\bm{\Sigma}}

\newcommand{\bPhi}{\bm{\Phi}}

\DeclareMathOperator{\diag}{diag}

\newtheorem{theorem}{Theorem}

\newcommand{\qednew}{\nobreak \ifvmode \relax \else
      \ifdim\lastskip<1.5em \hskip-\lastskip
      \hskip1.5em plus0em minus0.5em \fi \nobreak
      \vrule height0.75em width0.5em depth0.25em\fi}

\usepackage{amsthm}
\usepackage{xcolor, tikz}
\newtheorem{pro}{Proposition}

\newtheorem{lem}{Lemma}
\newtheorem{rem}{Remark}
\usepackage{bm}
\DeclareMathOperator{\sign}{sign}

\DeclareMathOperator{\rad}{rad}

\DeclareMathOperator{\dist}{dist}
\title{Phase Transitions in Phase-Only Compressed Sensing}
\author{Junren Chen\thanks{Department of Mathematics, The University of Hong Kong. email: \texttt{chenjr58@connect.hku.hk}}\and Lexiao Lai\thanks{Department of Mathematics, The University of Hong Kong. email: \texttt{lai.lexiao@hku.hk}} \and Arian Maleki\thanks{Department of Statistics, Columbia University. email: \texttt{mm4338@columbia.edu}}}
\date{\today}
\begin{document}
    \maketitle
\begin{abstract}
    The goal of phase-only compressed sensing is to recover a structured signal $\bx$ from the phases $\bz = \sign(\bPhi\bx)$, where $\bPhi$ is a {\it complex-valued} sensing matrix. As demonstrated in prior studies, exact reconstruction of the signal's direction is possible by reformulating the problem as a linear compressed sensing problem and applying basis pursuit (i.e., constrained norm minimization). For $\bPhi$ with i.i.d. complex-valued Gaussian entries, this paper demonstrates that the phase transition is approximately determined by the statistical dimension of the descent cone associated with a {\it signal-dependent norm}. Leveraging this insight, we derive asymptotically precise formulas for the phase transition locations in phase-only sensing of both sparse signals and low-rank matrices. Our results prove that the minimum number of measurements required for exact recovery is smaller for phase-only measurements than for traditional linear compressed sensing. For instance, in recovering a 1-sparse signal with sufficiently large dimension, phase-only compressed sensing requires approximately 68\% of the measurements needed for linear compressed sensing. This result disproves earlier conjecture suggesting that the two phase transitions coincide.
Our proof hinges on the Gaussian min-max theorem and the key observation that, up to a signal-dependent orthogonal transformation, the sensing matrix in the reformulated problem behaves as a nearly Gaussian matrix.
\end{abstract}
\section{Introduction}

\subsection{Main Challenges}
The recovery of structured signals from a limited number of nonlinear observations, referred to as {\it nonlinear compressed sensing}, has garnered significant research attention over the past decade; see e.g., \cite{plan2016generalized,plan2017high,xu2020quantized,genzel2023unified,chen2024unified,dirksen2019quantized,thrampoulidis2020generalized,dirksen2021non,chen2024optimal,chen2022quantizing,cai2016optimal,candes2015phase}. An important example within this area is 1-bit compressed sensing \cite{plan2013one,plan2012robust,jacques2013robust,matsumoto2024binary}, which studies the sensing mechanisms where the observations are quantized coarsely. Here, the goal is to recover a structured signal (e.g., a sparse vector) $\bx \in \mathbb{S}^{n-1} := \{\bu \in \mathbb{R}^n : \|\bu\|_2 = 1\}$ from $\by = \sign(\bA\bx)$, where $\bA \in \mathbb{R}^{m \times n}$ is the sensing matrix.

A generalization of 1-bit compressed sensing to the complex domain leads to another nonlinear model called phase-only compressed sensing (PO-CS) \cite{boufounos2013angle,boufounos2013sparse,feuillen2020ell,jacques2021importance,chen2023uniformpo,chen2024robust}. PO-CS studies the recovery of $\bx\in\mathbb{S}^{n-1}$ from $\bz=\sign(\bPhi\bx)$ where $\bPhi$ is a {\it complex-valued} sensing matrix. Here,  we let $\sign(c)=c/|c|$ for nonzero $c\in \mathbb{C}$ and conventionally set $\sign(0)=1$. Note that this function extracts the phase information of a complex number. PO-CS also generalizes the recovery of {\it unstructured} signals from the phases of complex measurements (i.e., phase-only measurements) studied in a line of earlier works; see \cite{oppenheim1981importance,chen2023signal} and the references therein.

%Let us focus on $\bPhi$ with i.i.d. $\calN(0,1)+\calN(0,1)\bi$ entries in the present paper. 
In stark contrast to 1-bit compressed sensing, exact signal reconstruction is achievable in PO-CS. This was first proved by Jacques and Feuillen \cite{jacques2021importance} through analyzing a {\it linearization} approach, as we restate here. Let $\Re(\cdot)$ and $\Im(\cdot)$ denote the real and imaginary parts of complex numbers respectively. It is straightforward to see that if $\bz=\sign(\bPhi\bx)$, then 
\begin{align}
    \frac{1}{\sqrt{m}}\Im(\diag(\bz^*)\bPhi)\bu= 0\label{part1mea}
\end{align} holds for any   $\bu\in \{\lambda\bx:\lambda>0\}$. To fix the scaling of $\bu$ and hopefully obtain a unique solution, we also add 
 \begin{align}
     \frac{1}{m}\Re(\bz^* \bPhi)\bu=1.\label{part2mea}
 \end{align} 
The two equalities mentioned in (\ref{part1mea}) and (\ref{part2mea}) enable us to transform PO-CS problem to a linear compressed sensing problem. More specifically, we can use  (\ref{part1mea}) and (\ref{part2mea}) and aim to solve the basis pursuit problem:
\begin{align}\label{1.2}
    \hat{\bx} = {\rm arg}\min~f(\bu),\quad{\rm s.t.~~}\bA_{\bz}\bu = \be_1,  
\end{align}
where $\bA_{\bz}$ is the ``new sensing matrix'' defined as 
\begin{align}\label{Azphi}
    \bA_{\bz} := \begin{bmatrix}
        \frac{1}{m}\Re(\bz^*\bPhi)\\
        \frac{1}{\sqrt{m}}\Im(\diag(\bz^*)\bPhi)
    \end{bmatrix}\in \mathbb{R}^{(m+1)\times n},
\end{align}
and $f(\cdot)$ is a norm that promotes certain structure (e.g., $\ell_1$ norm for recovering sparse signals). Once $\hat{\bx}$ is calculated, we   use $\bx^\sharp ={\hat{\bx}}/\|\hat{\bx}\|_2$ as the final estimate for $\bx$. An important feature of \eqref{Azphi} %which is important in our analysis 
is that $\bA_{\bz}$ depends on the true signal $\bx$. %This point will be discussed further in Section \ref{sec2}. \\

Jacques and Feuillen \cite{jacques2021importance} studied the  restricted isometry property (RIP) of $\bA_{\bz}$. In their study they assumed that $\bx \in \calK$, where $\calK$ is a symmetric cone. Then, they showed that for an i.i.d. complex Gaussian matrix $\bPhi$ and a fixed $\bx \in \mathbb{S}^{n-1}$, $\bA_{\bz}$ satisfies the RIP with distortion $\delta$ over  $\calK$ with high probability (w.h.p.), so long as the measurement number is at the order of $\delta^{-2}\omega^2((\calK-\mathbb{R}\bx)\cap\mathbb{S}^{n-1})$. Here $\omega(\calK)=\mathbbm{E}\sup_{\bu\in\calK}\bg^\top \bu$ with $\bg\sim\calN(0,\bI_n)$ denotes the Gaussian width of a set $\calK\subset \mathbb{R}^n$.  This implies that, for any fixed $s$-sparse $\bx$, the above recovery procedure with $f(\bu)=\|\bu\|_1$ succeeds
(i.e., $\bx^\sharp = \bx$) w.h.p. if the number of phases $m$ is at the order of $s\log(\frac{en}{s})$.
    In a follow-up work \cite{chen2023uniformpo}, the authors showed through a covering argument that $\bx^\sharp = \bx$ simultaneously holds for {\it all} $s$-sparse signals $\bx$ in $\mathbb{S}^{n-1}$ if $m\gtrsim s\log(n)$. More recently,  \cite{chen2024robust} established a stronger instance optimal guarantee  by a refined covering argument: for some absolute constants $C_1$ and $C_2$, if  $m\ge C_1s\log(\frac{en}{s})$, then w.h.p. we have    $$\|\bx^\sharp-\bx\|_2\le \frac{C_2 \min_{\bu\in\Sigma^n_s}\|\bu-\bx\|_1}{\sqrt{s}},\quad \forall \bx\in \mathbb{S}^{n-1},$$ 
  where $\Sigma^n_s$ denotes the set of all $s$-sparse  signals in $\mathbb{R}^{n}$. %We mention that  the recovery of complex signals and the robustness to various noise/corruption patterns were also established in \cite{chen2023uniformpo} and \cite{chen2024robust}, respectively. 
  It is important to note that all these results characterize only the order of $m$ needed for successful recovery and involve unspecified  multiplicative constants that are typically very large. This leads us to the following question that we aim to address in this paper:
  $$\textit{What~is~the~exact~number~of~phase-only~measurements~required~for~(\ref{1.2})~to~succeed?}$$    
Let us call this threshold $\zeta_{\rm PO}$. 
  As each phase-only observation can be viewed as a real scalar in $[0,2\pi)$, it is highly desired to compare $\zeta_{\rm PO}$ with the phase transition locations of   compressed sensing from $m$ {\it real-valued} linear measurements  via basis pursuit, i.e., the recovery of a structured $\bx\in\mathbb{S}^{n-1}$ from $\by = \bA\bx$ under $\bA\sim \calN^{m\times n}(0,1)$ through solving 
 \begin{align}
     \hat{\bx} = {\rm arg}\min~f(\bu),\quad {\rm s.t.}~~\bA\bu= \by.\label{linearbs}
 \end{align} 
 We denote the phase transition threshold of this problem by $\zeta_{\rm LN}$, which 
has been well-documented in, for instance, \cite{amelunxen2014living,donoho2006high,donoho2009counting,stojnic2010l,stojnic2013framework,donoho2009message, maleki2013asymptotic}. In \cite{jacques2021importance}, the authors reported an empirical observation that ``{\it PO-CS (from $\sign(\bPhi \bx)$) requires about twice the number of measurements required for perfect signal recovery in linear compressed sensing (from $\bPhi\bx$, i.e., $\Re(\bPhi)\bx$ and $\Im(\bPhi)\bx$)}'' \cite[P. 4154]{jacques2021importance}. Note that such a claim translates to $\zeta_{\rm PO}\approx \zeta_{\rm LN}$ in our notation. For unstructured signals, one can indeed prove $\zeta_{\rm PO}=n-1 \approx n = \zeta_{\rm LN}$ via non-probabilistic matrix analysis techniques \cite[Sec. 6]{chen2023signal}. This discussion leads us to the second question we aim to address in our paper:  
$$\textit{Can we rigorously prove or disprove $\zeta_{\rm PO}\approx \zeta_{\rm LN}$ for structured signals?}$$

\subsection{Our Contributions}  
The main aim of this paper is to establish the asymptotically exact number of measurements required for perfect recovery in (noiseless) PO-CS. Informally, we show that for a fixed signal $\bx\in\mathbb{S}^{n-1}$, there exists a threshold $\zeta_{\rm PO}(\bx;f)$ defined as 
  \begin{align}\label{zetapo} 
      \zeta_{\rm PO}(\bx;f) = \left[\mathbbm{E}\sup_{\substack{\bu\in \calT_f(\bx)\\ \|\bQ_{\bx}\bu\|_2= 1}}\big\langle (\bI_n-\bx\bx^\top)\bg,\bu\big\rangle\right]^2 
  \end{align}
  such that (\ref{1.2}) exactly recovers the direction of $\bx$ w.h.p. if $m\ge (1+o(1))\cdot\zeta_{\rm PO}(\bx;f)$, while the recovery fails (i.e., $\bx^\sharp\ne \bx$) w.h.p. if $m\le (1-o(1))\cdot\zeta_{\rm PO}(\bx;f)$. In (\ref{zetapo}), $\calT_f(\bx)=\{\bu\in \mathbb{R}^n:f(\bx+t\bu)\le f(\bx),~\exists t>0\}$ is the descent cone of $f$ at $\bx$, $\bQ_{\bx}:=\bI_n+(\sqrt{\frac{\pi}{2}}-1)\bx\bx^\top$ and $\bg \sim\calN(0,\bI_n)$. See Theorem \ref{mainthm} for the formal statement of this result. We establish this phase transition result by using the (convex) Gaussian min-max theorem (e.g., \cite{thrampoulidis2018precise}) that extends Gordon's Gaussian comparison inequality \cite{gordon1985some,gordon1988milman}, along with the key observation that $\bA_{\bz}$ is ``near-Gaussian'' when transformed by an orthogonal matrix.

  To address the second question we raised in the previous subsection, i.e. comparison of $\zeta_{\rm PO}$ and $\zeta_{\rm LN}$, we first establish \begin{align} \label{sdexpress}
      \zeta_{\rm PO}(\bx;f)\approx \delta\big(\calT_{f_{\bx}}(\bx)\big) 
  \end{align} in Proposition \ref{pro1}, where $f_{\bx}(\cdot)$ is a signal-dependent norm defined as $f_{\bx}(\bu):= f\big(\bu - (1-\sqrt{2/\pi})\langle \bx,\bu\rangle\bx\big)$, and $\delta(\calC) =  \mathbbm{E}\big[\sup_{\bu\in\calC \cap \mathbb{S}^{n-1}}\langle 
\bg,\bu \rangle\big]^2$ with  $\bg \sim \calN(0,\bI_n)$, denotes the statistical dimension of the closed convex cone $\calC$ in $\mathbb{R}^n$. This expression allows us to use the mechanisms proposed in \cite{amelunxen2014living} (for the calculation of statistical dimension) to derive explicit formulas for $\zeta_{\rm PO}(\bx;f)$. %Concerning the recovery of $\bx$, $f_{\bx}$ is generally preferable over the original $f$  in the sense that it penalizes signals being close to $\bx$ less (e.g., $f_{\bx}(\bx) = \sqrt{2/\pi}\cdot f(\bx)$ v.s. $f_{\bx}(\bu)=f(\bu)$ for $\bu$ being orthogonal to $\bx$). 
To compare $\zeta_{\rm PO}$ with $\zeta_{\rm LN}$, we should note that   the phase transition for the procedure (\ref{linearbs}) in linear compressed sensing is located at (see e.g. \cite{amelunxen2014living})
\begin{align}\label{zetaln}
    \zeta_{\rm LN}(\bx;f) = \delta\big(\calT_f(\bx)\big). 
\end{align} 
We will show that $\zeta_{\rm PO}(\bx;f)\le \zeta_{\rm LN}(\bx;f)$ and that the   observation $\zeta_{\rm PO}(\bx;f)\approx \zeta_{\rm LN}(\bx;f)$ made in \cite{jacques2021importance} is in fact {\it not} correct. Specifically, we derive asymptotically exact formulas for $\zeta_{\rm LN}(\bx;f)$ in the canonical cases of recovering a sparse vector via minimizing $\ell_1$ norm $f(\bu)=\|\bu\|_1$  and reconstructing a low-rank matrix $\bX$ via minimizing nuclear norm $f(\bU)=\|\bU\|_{nu}$,\footnote{The nuclear norm of a matrix is defined as the sum of its singular values.} referred to as sparse signal and low-rank matrix recovery respectively. In these two cases, we show that   $\zeta_{\rm PO}$ and $\zeta_{\rm LN}$ typically differ by a multiplicative factor. More precisely, $\zeta_{\rm PO}/\zeta_{\rm LN}$ is bounded away from $1$ as long as the sparsity or the rank is bounded away from the ambient dimension. %This gap becomes larger for more structured signals (e.g., sparser vectors, matrices of lower rank), while asymptotically vanishing for near unstructured signals (e.g., almost dense signal, matrix $\bX\in\mathbb{R}^{p\times q}$ with rank close to $\min\{p,q\}$). 

Figure \ref{fig1} shows that under fixed $n$, $\zeta_{\rm PO}(\bx;\ell_1)/\zeta_{\rm LN}(\bx;\ell_1)$ is monotonically increasing in the sparsity of the vector $s$ and tends to $1$ as $s\to n$. Some observation here already contradicts the numerical observation $\zeta_{\rm PO}\approx\zeta_{\rm LN}$   \cite{jacques2021importance}, e.g., it shows that $\zeta_{\rm PO}(\bx;\ell_1)<0.75\cdot\zeta_{\rm LN}(\bx;\ell_1)$ for $1$-sparse $\bx\in \mathbb{S}^{999}$. 
More precise statements of the ratio $\zeta_{\rm PO}/\zeta_{\rm LN}$ can be found in Section \ref{ratio} and Figure \ref{fig:ratio}.

\begin{figure}[ht!]
    \begin{centering}
        \includegraphics[width=0.55\columnwidth]{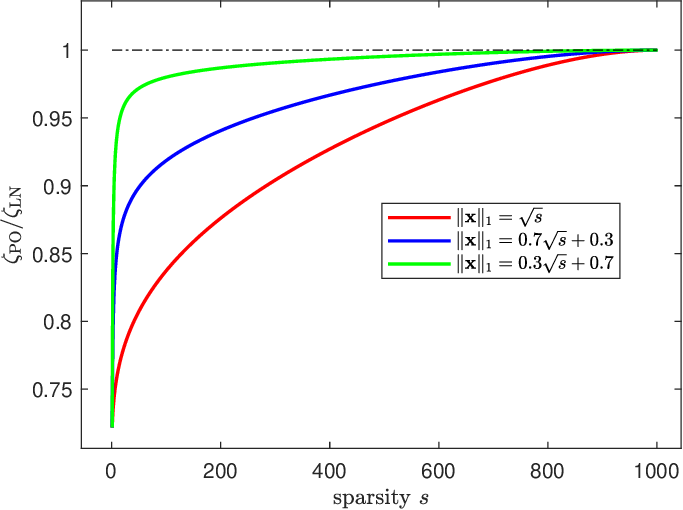} 
        \par
    \end{centering}
    
    \caption{We fix 
    $n=1000$ and plot the (approximate) curves of $\zeta_{\rm PO}/\zeta_{\rm LN}$ v.s. $s=1:1000$ under $\|\bx\|_1=\sqrt{s}$, $0.7\sqrt{s}+0.3,~0.3\sqrt{s}+0.7$. These curves are monotonically increasing. %In the right figure, we fix $s=1,~5,~20$ and plot the curves of $\zeta_{\rm PO}/\zeta_{\rm LN}$ v.s. $n=100:50:10^5$, which are monotonically decreasing. 
    \label{fig1}}
\end{figure}

As we will clarify later, there is another interesting and important difference between $\zeta_{\rm PO}$ and $\zeta_{\rm LN}$ in terms of the dependence on the signal. For example, while $\zeta_{\rm LN}(\bx;\ell_1)$ only depends on $\bx$ through the dimension $n$ and the sparsity $s$, $\zeta_{\rm PO}(\bx;\ell_1)$ also depends on $\|\bx\|_1$. Indeed, under fixed $n$ and $s\ge 2$, $\zeta_{\rm PO}(\bx;\ell_1)$ is a monotonically decreasing function in $\|\bx\|_1 \in (1,\sqrt{s}]$. Therefore, $s$-sparse signals with non-zero entries being $\pm s^{-1/2}$, which is referred to as the  {\it equal amplitude}  signal henceforth, renders the earliest phase transition in PO-CS. This is in stark contrast to a fact in linear compressed sensing that equal amplitude signal is somewhat the least favorable one in terms of the noise sensitivity \cite{donoho2011noise}.

%And  
 %as we shall see, the substitution of $\bg$ in (\ref{zetaln}) with $(\bI_n-\bx\bx^\top)\bg$ in (\ref{zetapo}) only induces relatively minor difference, while the replacement of $\bu\in \mathbb{S}^{n-1}$ (implicitly from $\bu\in\calT_f^*(\bx)$ in (\ref{zetaln})) by $\|\bQ_{\bx}\bu\|_2=1$  in (\ref{zetapo}) appears more notable and often creates a gap of distinct multiplicative factors, informally: $$\zeta_{\rm LN}(\bx;f)-\zeta_{\rm PO}(\bx;f)\gtrsim \zeta_{\rm LN}(\bx;f).$$ 
 %Another interesting finding is that $\zeta_{\rm PO}(\bx;f)$ typically exhibits more dependence on the underlying signal then $\zeta_{\rm LN}(\bx;f)$.

\subsection{Notation}
 Throughout this paper, $f$ is a properly chosen norm in $\mathbb{R}^n$ (hence $f(\lambda \bu)=|\lambda|f(\bu)$), $\bx\in \mathbb{S}^{n-1}$ is a fixed underlying signal. The descent cone of $f$  at $\bw$ is given by $\calT_f(\bw) = \{\bu \in\mathbb{R}^n:\exists t>0,~f(\bw+t\bu)\le f(\bw)\}$, and the subdifferential of $f$ at $\bw$ is defined as $  \partial f(\bw) =\{\bu\in\mathbb{R}^n:f(\by)\ge f(\bw)+\langle \bu, \by-\bw\rangle,~\forall \by\in \mathbb{R}^n\}.$ For $\ba\in \mathbb{R}^n$, we work with the $\ell_2$-norm $\|\ba\|_2$, and $\ell_1$-norm $\|\ba\|_1$. For matrix $\bA \in\mathbb{R}^{p\times q}$, we let $\|\bA\|_{op},~\|\bA\|_{nu},~\|\bA\|_{F}$ denote the operator norm (the largest singular value),   the nuclear norm, and the Frobenius norm, respectively.   Given $\calW\subset \mathbb{R}^n$, we define its radius as $\rad(\calW)= \sup\{\|\bw\|_2:\bw\in\calW\}$ and the distance from $\ba$ to $\calW$ as $\dist(\ba,\calW)=\inf_{\bw\in\calW}\|\ba-\bw\|_2$. We write the inner product of two matrices (or vectors if $q= 1$) $\bA,\bB\in \mathbb{R}^{p\times q}$ as $\langle\bA,\bB\rangle = {\rm Tr}(\bA^\top \bB)$.   The dual of a cone $\calC$ is defined as $\calC^\circ=\{\bu\in \mathbb{R}^n:\langle\bu,\bw\rangle\le 0,~~\forall \bw\in\calC\}$. The statistical dimension of a closed convex cone $\calC$    is denoted by $\delta(\calC)$ \cite[Prop. 3.1]{amelunxen2014living} and is  closely related to the Gaussian width: for a closed convex cone $\calC$ we have $\omega^2(\calC^*)\le \delta(\calC)\le \omega^2(\calC^*)+1$ \cite[Prop. 10.2]{amelunxen2014living} where we   use the convention $\calC^*=\calC\cap \mathbb{S}^{n-1}$. For two random variables, we write $A\stackrel{d}{=} B$ to state that $A$ and $B$ have the same distribution. We also use $\calN^{p\times q}(0,1)$ to denote a $p\times q$ random matrix with i.i.d. $\calN(0,1)$ entries. Throughout the paper, positive constants are denoted with $C,C_i,c,c_i$ whose values may vary from line to line. We   write $T_1\lesssim T_2$ ($T_1\gtrsim T_2$, resp.) if there exists some positive constant $C$  such that $T_1\le CT_2$ ($T_1\ge CT_2$, resp.).

\subsection{Structure of the Paper}
We present our main result which pinpoints the phase transition location of PO-CS in Section \ref{sec2}. In Section \ref{sec3} we calculate $\zeta_{\rm PO}(\bx;f)$ in sparse signal and low-rank matrix recovery, respectively and approximately plot the ratios $\zeta_{\rm PO}/\zeta_{\rm LN}$. We provide numerical results in Section \ref{sec4} and conclude the paper in Section \ref{sec5}.  The complete technical proofs are deferred to the appendix.

\section{Phase Transition}\label{sec2}
Recall that we consider the following recovery procedure: we solve $\hat{\bx}$ from (\ref{1.2}) with  $\bA_{\bz}$ given by (\ref{Azphi}) and then normalize it to obtain $\bx^\sharp={\hat{\bx}}/\|\hat{\bx}\|_2$ as our final estimate. First, we need to find the sufficient conditions that account for the success and failure of this procedure. Fortunately, despite the additional normalization in Euclidean norm, we still have the following statement similar to the ones in \cite{chandrasekaran2012convex,amelunxen2014living,sun2022phase,oymak2018universality} for linear sensing. 
\begin{lem}[Conditions for success and failure]\label{lem:suffcon}
    For a fixed $\bx\in\mathbb{S}^{n-1}$, suppose that $\|\bPhi\bx\|_1>0$. Let $f(\cdot)$ denote a norm in $\mathbb{R}^n$, and assume that $\calT_f(\bx)$ is closed and not a subspace. Then the following two statements are correct: 
    \begin{itemize}
        \item We have $\bx^\sharp = \bx$ if  
        \begin{align}\label{succ_con}
            \min_{\bu\in \calT_f^*(\bx)}\max_{\bv\in\mathbb{S}^m}~\bv^\top \bA_{\bz}\bu>0. 
        \end{align}
  
        \item We have $\bx^\sharp\ne \bx$ if  
        \begin{align}\label{fail_con}
            \min_{\bv\in \mathbb{S}^m}\max_{\bu\in\calT_f^*(\bx)}~\bv^\top\bA_{\bz}\bu>0. 
        \end{align} 
    \end{itemize}
\end{lem}
 
  See Appendix \ref{provelem1} for the proof of Lemma \ref{lem:suffcon}.   We aim to use the conditions in Lemma \ref{lem:suffcon} to establish the phase transition threshold of (\ref{1.2}). Towards this goal, we use the (convex) Gaussian mix-max theorem (see, e.g., \cite{thrampoulidis2015regularized,thrampoulidis2018precise}), particularly Lemma \ref{gmt} below that follows from  the more general statement in \cite[Prop. 1]{chandrasekher2023sharp}. 
  \begin{lem}\label{gmt}
    We let $\bG\in \mathbb{R}^{m\times (n-1)}$, $\bg\in \mathbb{R}^m$, $\bh\in \mathbb{R}^{n-1}$, $\calS_{\bw}\subset \mathbb{R}^n$, $\calS_{\bu}\subset \mathbb{R}^m$, $\psi:\mathbb{R}^n\times \mathbb{R}^m\to\mathbb{R}$, $\bw=(w_1,\tilde{\bw}^\top)^\top \in \calS_{\bw}$ with $\tilde{\bw}\in \mathbb{R}^{n-1}$, and define 
    \begin{gather*}
        \Gamma(\bG) := \min_{\bw\in\calS_{\bw}}\max_{\bu\in\calS_{\bu}}~\bu^\top\bG\tilde{\bw}+ \psi(\bw,\bu),\\
        \Upsilon(\bg,\bh) := \min_{\bw\in\calS_{\bw}}\max_{\bu\in\calS_{\bu}}~\|\tilde{\bw}\|_2\bg^\top\bu+\|\bu\|_2\bh^\top\tilde{\bw}+\psi(\bw,\bu).
    \end{gather*}
    Assume that $\calS_{\bw}$ and $\calS_{\bu}$ are compact, $\psi$ is continuous on $\calS_{\bw}\times\calS_{\bu}$, and $\bG,\bg,\bh$ have i.i.d. standard normal entries. Then for any $c\in\mathbb{R}$ we have $\mathbbm{P}(\Gamma(\bG)<c)\le 2 \mathbbm{P}(\Upsilon(\bg,\bh)\le c)$, or equivalently $\mathbbm{P}(\Gamma(\bG)\ge c)\ge 2 \mathbbm{P}(\Upsilon(\bg,\bh)>c)-1$
\end{lem}

 Note that Lemma \ref{gmt} does not directly apply to 
(\ref{succ_con}) or (\ref{fail_con}), as it requires the matrix to have i.i.d. $\mathcal{N}(0,1)$ entries, which is not the case for $\bA_{\bz}$. Our next Lemma aims to show that $\bA_{\bz}$ is ``close'' to an i.i.d. Gaussian matrix in some sense. 
Throughout this work, $\bP_{\bx}\in \mathbb{R}^{n\times n}$ denotes an orthogonal matrix whose first row is $\bx^\top$. Note that $\bP_{\bx}\bx = \be_1$. %To allow for subsequent use of Gaussian min-max theorem, the next lemma shows that  $\bA_{\bz}$ is ``close'' to standard Gaussian matrix when transformed by $\bP_{\bx}^\top$. 
\begin{lem}
    [Nearly-Gaussianity of $\bA_{\bz}$] \label{lem2}Let $L=\frac{1}{m}\sum_{i=1}^mL_i$ where $\{L_i\}_{i=1}^m\stackrel{iid}{\sim}|\calN(0,1)+\calN(0,1)\bi|$, and let $\bG\sim \calN^{(m+1)\times (n-1)}(0,1)$ be independent of $L$. Then we have 
    \begin{align}
        \bA_{\bz}\bP_\bx^\top \stackrel{d}{=} \begin{bmatrix}
            L\be_1 & \frac{\bG}{\sqrt{m}}
        \end{bmatrix}.
    \end{align}
\end{lem}
See Appendix \ref{provelem2} for the detailed proof.  Using this lemma, we are ready to state our main theorem and give a sketch of the proof.    
%Our main technical tool is the following Gaussian min-max theorem (see, e.g., \cite{thrampoulidis2015regularized,thrampoulidis2018precise}) that follows from the more general statement in \cite[Prop. 1]{chandrasekher2023sharp}. 
 
\begin{theorem}\label{mainthm}
     Suppose that  the entries of $\bPhi$ are i.i.d. $\calN(0,1)+\calN(0,1)\bi$, and consider the recovery of a fixed signal $\bx\in \mathbb{S}^{n-1}$ from $\bz=\sign(\bPhi\bx)$  by first solving $\hat{\bx}$ from (\ref{1.2}) and then obtaining $\bx^\sharp = \hat{\bx}/\|\hat{\bx}\|_2$ as the final estimate. There  exists an absolute constant $c$, such that for any $t\in[\frac{17}{m},c]$, we have:
    \begin{itemize}
        \item If $m\ge (1+t)\cdot \zeta_{\rm PO}(\bx;f)$, then $\mathbbm{P}(\bx^\sharp=\bx) \ge 1-14\exp(-\frac{mt^2}{289})$;
        \item If $m\le (1-t)\cdot\zeta_{\rm PO}(\bx;f)$, then $\mathbbm{P}(\bx^\sharp\ne \bx) \ge 1-14\exp(-\frac{mt^2}{289})$,
    \end{itemize}
    where $\zeta_{\rm PO}(\bx;f)$ is the quantity defined in (\ref{zetapo}).
\end{theorem}
\begin{proof}[Proof Sketch] 
    We defer the complete proof to Appendix \ref{provethm1} and mention some important steps/ideas here. We only mention the proof for the condition of successful recovery. The arguments for the failure case are similar. By the first statement in Lemma \ref{lem:suffcon}, we have that  
    \begin{align}\label{eq:first_step}
    \mathbbm{P}(\bx^\sharp =\bx)\ge \mathbbm{P}\big(\min_{\bu\in\calT_f^*(\bx)}\max_{\bv\in\mathbb{S}^m}\bv^\top\bA_{\bz}\bu>0\big). 
    \end{align}
The minimax form we see inside the probability on the right hand side of \eqref{eq:first_step} is similar to the forms  Lemma \ref{gmt} can simplify. However, $\bA_{\bz}$ is not an i.i.d. Gaussian matrix. Hence, to make Lemma \ref{gmt} applicable, we use Lemma  \ref{lem2} to simplify $\bA_{\bz}$. According to Lemma \ref{lem2} we have  $\bv^\top \bA_{\bz}\bu = \bv^\top \bA_{\bz}\bP_{\bx}^\top \bP_{\bx}\bu \stackrel{d}{=} \bv^\top \big[L\be_1~\frac{\bG}{\sqrt{m}}\big](\bP_{\bx}\bu)$, where $\bG$ has i.i.d.~Gaussian entries. Hence, the idea is to apply Lemma \ref{gmt} to matrix $\bG$. Towards this goal, we first condition on $L = \frac{1}{m}\sum_{i=1}^m L_i$  (with $\{L_i\}_{i=1}^m\stackrel{iid}{\sim}|\calN(0,1)+\calN(0,1)\bi|$) and use Lemma \ref{gmt} regarding the randomness of $\bG$ to reach the lower bound 
    \begin{align}\label{lower1}
       & \mathbbm{P}(\bx^\sharp =\bx) \ge 2 \mathbbm{P}\Big(\forall \bu \in\bP_{\bx}\calT_f^*(\bx),~\bh^\top \tilde{\bu}\le \Big\|\|\tilde{\bu}\|_2\bg + \sqrt{m}Lu_1\be_1\Big\|_2\Big)-1 , 
    \end{align}
    where $\bg \sim\calN(0,\bI_{m+1})$ is independent of $L$, and $\tilde{\bu}$ is defined through $\bu=(u_1,\tilde{\bu}^\top)^\top$. To see more details on how \eqref{lower1} is derived, the reader can refer to Appendix \ref{provethm1}. Here we simply accept \eqref{lower1} and present the rest of the ideas.

    To obtain a lower bound for the right hand side of \eqref{lower1} , an important observation is that $\|\|\tilde{\bu}\|_2\bg+\sqrt{m}Lu_1\be_1\|_2$ tightly concentrates around $\big[m(\|\tilde{\bu}\|_2^2+\frac{\pi u_1^2}{2})\big]^{1/2}$; hence, according to (\ref{lower1}), $\mathbbm{P}(\bx^\sharp=\bx)$ is approximately bounded from below by
    $2\mathbbm{P}$ $\Big($ $\sup_{\bu\in \bP_{\bx}\calT^*_f(\bx)}$$ \big\{\bh^\top\tilde{\bu}/\big(\|\tilde{\bu}\|_2^2+\frac{\pi u_1^2}{2}\big)^{1/2}\big\}\le \sqrt{m}\Big)-1.$
    Moreover, $\sup_{\bu\in \bP_{\bx}\calT^*_f(\bx)} \big\{\bh^\top\tilde{\bu}/\big(\|\tilde{\bu}\|_2^2+\frac{\pi u_1^2}{2}\big)^{1/2}\big\}$ is a $1$-Lipschitz function of the Gaussian vector $\bh$ and hence tightly concentrates around its expectation (e.g., see \cite[Eq. (1.6)]{ledoux2013probability}) that is equal to $\sqrt{\zeta_{\rm PO}(\bx;f)}$. This can be shown by some algebra. Therefore, $\sup_{\bu\in \bP_{\bx}\calT^*_f(\bx)} \big\{\bh^\top\tilde{\bu}/\big(\|\tilde{\bu}\|_2^2+\frac{\pi u_1^2}{2}\big)^{1/2}\big\}\le \sqrt{m}$ occurs with high probability if $m$ exceeds $\zeta_{\rm PO}(\bx;f)$.  
 \end{proof}

An immediate observation is that the phase transition for PO-CS always occurs at lower values of \(m\) compared to linear compressed sensing. This can be understood by noting that  
\begin{align*}
        \mathbbm{E}\sup_{\substack{\bu\in \calT_f(\bx)\\ \|\bQ_{\bx}\bu\|_2= 1}}\big\langle (\bI_n-\bx\bx^\top)\bg,\bu\big\rangle \le \mathbbm{E}\sup_{\bu\in\calT_f^*(\bx)}\big\langle (\bI_n-\bx\bx^\top)\bg,\bu\big\rangle = \omega\big((\bI_n-\bx\bx^\top)\calT_f^*(\bx)\big) \le \omega(\calT_f^*(\bx)). 
    \end{align*} 
    This implies that  
    \begin{align*}
    \zeta_{\rm PO}(\bx;f)\le  \omega^2(\calT^*_f(\bx))\le \delta(\calT_f(\bx))=\zeta_{\rm LN}(\bx;f).
   \end{align*} 
     In Section \ref{sec3}, we will show for sparse/low-rank recovery that the ratio $\zeta_{\rm PO}(\bx;f)/\zeta_{\rm LN}(\bx;f)$ is typically bounded away from $1$ for signals with ``non-trivial'' structure. The complex form of \(\zeta_{\rm PO}(\bx; f)\) makes it challenging to compute directly or to perform accurate comparisons with linear compressed sensing. To address this issue, in our next proposition, we provide a simpler approximation of \(\zeta_{\rm PO}(\bx; f)\) based on the statistical dimension of a descent cone. Recall that
$\bQ_{\bx}=\bI_n+(\sqrt{\pi/2}-1)\bx\bx^\top$ and hence $\bQ_{\bx}^{-1}=
     \bI_n- (1-\sqrt{2/\pi})\bx\bx^\top$.

\begin{pro}\label{pro1}
    If we define the signal-dependent norm $f_{\bx}(\bw) = f(\bQ_{\bx}^{-1}\bw)$, then we  have 
     \begin{align*}
      \delta\big(\calT_{f_\bx}(\bx)\big)-\sqrt{\frac{8\delta(\calT_{f_{\bx}}(\bx))}{\pi}} -\Big(1 -\frac{2}{\pi}\Big)\le \zeta_{\rm PO}(\bx;f)\le \delta \big(\calT_{f_{\bx}}(\bx)\big).
     \end{align*}
\end{pro}
     The proof of Proposition \ref{pro1} can be found in Appendix \ref{provepro1}. In the following,  $\delta(\calT_{f_{\bx}}(\bx))$ will be used as a reliable surrogate for $\zeta_{\rm PO}(\bx;f)$. This is because the term $1-2/\pi+\sqrt{8\delta(\calT_{f_{\bx}}(\bx))/\pi}$ is much smaller than $\zeta_{\rm PO}(\bx;f)$. Hence, Proposition \ref{pro1} leads us to the following interesting remark. 
   \begin{rem}
      We have $\zeta_{\rm PO}(\bx;f) \approx \delta\big(\calT_{f_{\bx}}(\bx)\big) = \zeta_{\rm LN}(\bx;f_{\bx})$. Thus, using PO-CS to recover $\bx\in \mathbb{S}^{n-1}$ with norm $f$ needs nearly the same number of measurements as using linear compressed sensing via basis pursuit with the signal-dependent norm $f_{\bx}$.   
   \end{rem}    

\section{PO-CS of Sparse Signal and Low-Rank Matrix}\label{sec3}
To provide a clearer comparison between \(\zeta_{\rm PO}\) and \(\zeta_{\rm LN}\), we focus on two specific signal structures: sparse vectors and low-rank matrices. Using Proposition \ref{pro1}, we derive asymptotically exact formulas for \(\zeta_{\rm PO}(\bx; f)\) in these cases. The following theorem presents a simplified approximation for \(\zeta_{\rm PO}\).

\begin{theorem}  \label{thm2}  Define the  following surrogate for $\zeta_{\rm PO}(\bx;f)$
    \begin{align*}
        \hat{\zeta}_{\rm PO}(\bx;f) := \inf_{\tau\ge 0}~\mathbbm{E}\Big[\dist^2(\bg, \tau\cdot\bQ_{\bx}^{-1}\partial f(\bx))\Big]
    \end{align*}
    with $\bg\sim \calN(0,\bI_n)$. Then we have  
    \begin{align*}
      \hat{\zeta}_{\rm PO}(\bx;f)  -  \Big(\frac{8\hat{\zeta}_{\rm PO}(\bx;f) }{\pi}\Big)^{1/2} - \frac{\sqrt{2\pi}\rad(\bQ_{\bx}^{-1}\partial f(\bx))}{f(\bx)}- \Big(1-\frac{2}{\pi}\Big)\le \zeta_{\rm PO}(\bx;f) \le \hat{\zeta}_{\rm PO}(\bx;f).
    \end{align*}
\end{theorem}
 Proposition \ref{pro1} allows us to use  $\delta(\calT_{f_{\bx}}(\bx))$ as a surrogate for $\zeta_{\rm PO}(f;\bx)$, and to prove Theorem \ref{thm2}, all that remains is to compute $\delta(\calT_{f_{\bx}}(\bx))$ via the general recipe developed in \cite{amelunxen2014living}.  See Appendix \ref{provethm2}.
%In the following, we will establish explicit formulas for $\hat{\zeta}_{\rm PO}(\bx;f)$ in the recovery and low-rank recovery. 
We are now in a position to derive explicit formulas for $\zeta_{\rm PO}(\bx;f)$ in sparse   and low-rank recovery problems. 

\subsection{Sparse   Recovery}
We consider the recovery of $\bx\in \mathbb{S}^{n-1}$ with $s$ non-zero entries via basis pursuit with the $\ell_1$-norm $f(\bw)=\|\bw\|_1$. We remind the reader that our goal is to compare the phase transition of the linear compressed sensing problem with that of PO-CS. 

For the linear compressed sensing problem, we generate  $\bA\sim \calN^{m\times n}(0,1)$ to obtain $\by=\bA\bx$, and then solve basis pursuit ($\min~\|\bw\|_1,~~{\rm s.t. 
 }~~\bA\bw=\by$) to obtain an estimate. The phase transition is located at $\zeta_{\rm LN}(\bx;\|\cdot\|_1)$ that can be well approximated by (see for instance \cite[Prop. 4.5]{amelunxen2014living})
\begin{align}\label{eq:zeta_LN_sparse}
    \hat{\zeta}_{\rm LN}(\bx;\|\cdot\|_1) := n\psi_1 \Big(\frac{s}{n}\Big),\quad{\rm where~} \psi_1(u):= \inf_{\tau\ge 0}~\Big\{u(1+\tau^2)+(1-u)\sqrt{\frac{2}{\pi}}\int_{\tau}^\infty (w-\tau)^2 e^{-\frac{w^2}{2}}~\text{d}w\Big\}.  
\end{align}
%It is easy to show that, for any $u\in (0,1)$ we have
%\begin{align}
 %   \psi_1(u) = u\big(1+(\tau_1^*)^2\big) + (1-u)\sqrt{2/\pi}\int_{\tau_1^*}^\infty (w-\tau_1^*)^2\exp\big(-\frac{w^2}{2}\big)~\text{d}w
%\end{align}
%where $\tau_1^*>0$ is uniquely determined by the equation 
%\begin{align}
 %   \int_{\tau_1^*}^\infty \big(\frac{w}{\tau_1^*}-1\big)\exp\big(-\frac{w^2}{2}\big)~\text{d}w = \sqrt{\frac{\pi}{2}}\frac{u}{1-u}. 
%\end{align}

In the phase-only sensing scenario, we generate a complex-valued   matrix $\bPhi \sim \calN^{m\times n}(0,1)+ \calN^{m\times n}(0,1)\bi$ to observe the phases $\bz= \sign(\bPhi\bx)$, solve $\hat{\bx}$ from (\ref{1.2})--(\ref{Azphi}) with $f$ being $\ell_1$-norm, and then use $\bx^\sharp = \hat{\bx}/\|\hat{\bx}\|_2$ as the final estimate. The following result calculates $\hat{\zeta}_{\rm PO}(\bx;\|\cdot\|_1)$. 
\begin{theorem} \label{spa_formu}
    In the sparse recovery setting described above, we have 
    \begin{align}\label{61}
        \hat{\zeta}_{\rm PO}(\bx;\|\cdot\|_1) =  n\psi\Big(\frac{s}{n},\frac{\|\bx\|_1^2}{s}\Big),
    \end{align}
    where for any $(u,v)\in(0,1)\times (0,1]$,  
    \begin{align}\label{62}
        \psi(u,v):&=\inf_{\tau\ge 0}~\Big\{u\Big(1+\tau^2-\tau^2v (1-\frac{2}{\pi})\Big) + (1-u)\sqrt{\frac{2}{\pi}}\int_{\tau}^\infty (w-\tau)^2 e^{-\frac{w^2}{2}}~\text{d}w\Big\}.
    \end{align}
    Consequently, we have
    \begin{align}\label{63}
     n\psi \Big(\frac{s}{n},\frac{\|\bx\|_1^2}{s}\Big) -\Big(\frac{8n\psi(\frac{s}{n},\frac{\|\bx\|_1^2}{s})}{\pi}\Big)^{1/2} - \frac{\sqrt{2\pi n}}{\|\bx\|_1} -\Big(1-\frac{2}{\pi}\Big)  \le \zeta_{\rm PO}(\bx;f)\le n\psi \Big(\frac{s}{n},\frac{\|\bx\|_1^2}{s}\Big). 
    \end{align}
\end{theorem}
    The proof can be found in Appendix \ref{provethm3}. Compared to $\zeta_{\rm LN}(\bx;\|\cdot\|_1)$ and its surrogate $n\psi_1(\frac{s}{n})$ that only depend on $(n,s)$, we note that the surrogate threshold $\hat{\zeta}_{\rm PO}(\bx,\|\cdot\|_1)$ for PO-CS of sparse vector depends on $\|\bx\|_1$ too. We would like to make the following two remarks that clarify some of the interesting features of the phase transition of PO-CS. 

    \begin{rem}
    Since $\|\bx\|_2 =1$, it is straightforward to see that  $\|\bx\|_1=\sqrt{s}$, which corresponds to equal amplitude $s$-sparse signals, leads to the smallest $n\psi(\frac{s}{n},\frac{\|\bx\|_1^2}{s})$. In other words, in terms of the phase transition, the equal amplitude signal is the `most favorable signal' for PO-CS. While in the linear compressed sensing the amplitude of the non-zero coefficients do not have any impact on the location of the phase transition; in the noisy linear compressed sensing, as shown in \cite{donoho2011noise}, the equal amplitude signal is in fact the least favorable distribution.
    \end{rem}
    
\begin{rem}
The expression $n\psi(\frac{s}{n},\frac{\|\bx\|_1^2}{s})$ gives an asymptotically exact formula for the phase transition of PO-CS in the sense that $ \frac{\zeta_{\rm PO}(\bx;f)}{n} \to \psi\big(\frac{s}{n},\frac{\|\bx\|_1^2}{s}\big)$ as  $n\to \infty$.    
 \end{rem}

\subsection{Low-Rank   Recovery}
The aim of this section is to study the recovery of a rank-$r$ matrix $\bX\in\mathbb{R}^{p\times q}$ ($r \le p\le q$) with unit Frobenius norm.  Similar to the previous subsection, to study the linear compressed sensing, we generate sensing matrices $\{\calA_i\}_{i=1}^m$ with i.i.d. $\calN(0,1)$ entries to observe $\{y_i=\langle \calA_i,\bX\rangle\}_{i=1}^m$, and then solve basis pursuit ($\min~\|\bU\|_{nu},~~{\rm s.t.~~}y_i=\langle\calA_i,\bU\rangle,~\forall i \in [m]$) to obtain the estimate. An    asymptotically exact formula for the phase transition threshold $\zeta_{\rm LN}(\bX;\|\cdot\|_{nu})$ is given by $pq\Psi_1(\frac{r}{p},\frac{p}{q})$ (see for instance \cite[Prop. 4.7]{amelunxen2014living}):  
\begin{align}\label{eq:zeta_LN}
    \frac{\zeta_{\rm LN}(\bX;\|\cdot\|_{nu})}{pq} \to \Psi_1(\rho,\nu)
\end{align}
as $r,p,q\to\infty$  with limiting ratios $\frac{r}{p}\to \rho \in(0,1)$ and $\frac{p}{q}\to \nu\in (0,1]$, where $\Psi_1(\rho,\nu)$ is defined as 
\begin{gather}\label{Psi1mintau}
    \Psi_1(\rho,\nu):= \inf_{\tau\ge 0}~\Big\{\rho\nu+(1-\rho\nu)\Big[\rho(1+\tau^2)+(1-\rho)\int_{\max\{a_-,\tau\}}^{a_+}(b-\tau)^2\varphi_{y}(b)~\text{d}b\Big]\Big\}
\end{gather}
with $y = \frac{\nu-\rho\nu}{1-\rho \nu}$, $a_{\pm}=1\pm\sqrt{y}$, and 
\begin{align}\label{phiyb}
    \varphi_y(b) = \frac{1}{\pi yb}\sqrt{(b^2-a_-^2)(a_+^2-b^2)},\quad b\in [a_-,a_+]
\end{align}
is a probability density supported on $[a_-,a_+]$. %Moreover, for any $(\rho,\nu)\in(0,1)\times(0,1]$, the $\tau$ such that (\ref{Psi1mintau}) attains the minimum is the unique $\tau_1^*>0$ satisfying 
 %\begin{align}
  %      \int^{a_+}_{\max\{a_-,\tau_1^*\}}\big(\frac{b}{\tau^*_1}-1\big)\varphi_y(b)~\text{d}b = \frac{\rho}{1-\rho}.
   % \end{align}

In  PO-CS, we generate complex-valued sensing matrices $\{ \bPhi_i\}_{i=1}^m$ with i.i.d. $\calN(0,1)+\calN(0,1)\bi$ entries to observe the phases $\{y_i=\sign(\langle \bPhi_i,\bX\rangle)\}$, and then we can use the procedure described in (\ref{1.2})--(\ref{Azphi}) with $f$ being the nuclear norm to recover the underlying matrix. The following theorem provides an asymptotically exact formula for $ \zeta_{\rm PO}(\bX;\|\cdot\|_{nu})$, suggesting $pq\Psi(\frac{r}{p},\frac{p}{q},\frac{\|\bX\|_{nu}^2}{r})$ as a surrogate. 

\begin{theorem}\label{lowrankfor}
    In the above-described setting of PO-CS of a rank-$r$ matrix $\bX\in \mathbb{R}^{p\times q}$, suppose $r,p,q\to\infty$ with limiting ratios $\frac{r}{p}\to \rho\in(0,1)$ and $\frac{p}{q}\to \nu\in (0,1]$, then we have 
    \begin{align}\label{eq:zeta_PO}
        \frac{{\zeta}_{\rm PO}(\bX;\|\cdot\|_{nu})}{pq} \to \Psi\Big(\rho,\nu,\frac{\|\bX\|_{nu}^2}{r}\Big)
    \end{align}
    where $\Psi(\rho,\nu,\mu)$ is defined as 
    \begin{align}
        \Psi(\rho,\nu,\mu) := \inf_{\tau\ge 0}~\Big\{\rho\nu+(1-\rho\nu)\Big[\rho\Big(1+\Big[1-(1-\frac{2}{\pi})\mu\Big]\tau^2\Big)+(1-\rho)\int_{\max\{a_-,\tau\}}^{a_+}(b-\tau)^2\varphi_{y}(b)~\mathrm{d}b\Big]\Big\}, \label{inftau}
    \end{align}
    with $y= \frac{\nu-\rho \nu}{1-\rho\nu}$, $a_{\pm}= 1\pm\sqrt{y}$ and $\varphi_y(b)$ defined in (\ref{phiyb}). %For $\rho\in (0,1)$ and $\nu,\mu\in(0,1]$, the    $\tau$ such that (\ref{inftau}) attains the minimum is the unique $\tau^*>0$ satisfying 
  %  \begin{align}\label{taustar}
   %     \int^{a_+}_{\max\{a_-,\tau^*\}}\big(\frac{b}{\tau^*}-1\big)\varphi_y(b)~\text{d}b = \frac{\rho(1-(1-2/\pi)\mu)}{1-\rho}.
   % \end{align}
\end{theorem}
      The proof of this theorem can be found in Appendix \ref{provethm4}.

    Compared to $\zeta_{\rm LN}(\bX;\|\cdot\|_{nu})$ and its surrogate $pq\Psi_1(\rho,\nu)$ that depend only on $(p,q,r)$, we find that $\zeta_{\rm PO}(\bX;\|\cdot\|_{nu})$'s surrogate  also depends on $\|\bX\|_{nu}$, and evidently $\|\bX\|_{nu}=\sqrt{r}$ leads to the earliest phase transition, in the sense of the smallest $pq\cdot \Psi(\frac{r}{p},\frac{p}{q},\frac{\|\bX\|_{nu}^2}{r})$.

\subsection{Simulating  $\zeta_{\rm PO}/\zeta_{\rm LN}$} \label{ratio}
In this section, we   simulate the ratio $\frac{\zeta_{\rm PO}}{\zeta_{\rm LN}}$ by plotting the curves of $\frac{\hat{\zeta}_{\rm PO}}{\hat{\zeta}_{\rm LN}}$. For sparse signal recovery, we can use 
$\frac{\hat{\zeta}_{\rm PO}(\bx;\|\cdot\|_1)}{\hat{\zeta}_{\rm LN}(\bx;\|\cdot\|_1)} = {\psi(\frac{s}{n},\frac{\|\bx\|_1^2}{s})}/{\psi_1(\frac{s}{n})}
$
 as a surrogate; similarly, we can use 
$
     {\Psi(\frac{r}{p},\frac{p}{q},\frac{\|\bX\|_{nu}^2}{r})}/\Psi_1(\frac{r}{p},\frac{p}{q})
$ to  approximate the ratio in low-rank matrix recovery. Therefore, we are interested in the two ratio functions 
\begin{gather}
    R_{\rm sp}(u,v) = \frac{\psi(u,v)}{\psi_1(u)},\quad (u,v)\in(0,1] ^2,\\
    R_{\rm lr}(u,v,w) = \frac{\Psi(u,v,w)}{\Psi_1(u,v)},\quad (u,v,w)\in (0,1]^3. 
\end{gather}
Given $v,w\in (0,1]$, we find that $R_{\rm sp}(u,v)$ and $R_{\rm lr}(u,v,w)$ are monotonically increasing in $u\in (0,1]$ and they both converge to $1$ as $u\to 1^-$. See Figure \ref{fig:ratio} and the associated texts for details. We should just mention that an analytical analysis of $R_{\rm sp}$ and $R_{\rm lr}$ is not pursued in the present paper. 
\begin{figure}[!ht]
    \centering
    \includegraphics[width=0.4\linewidth]{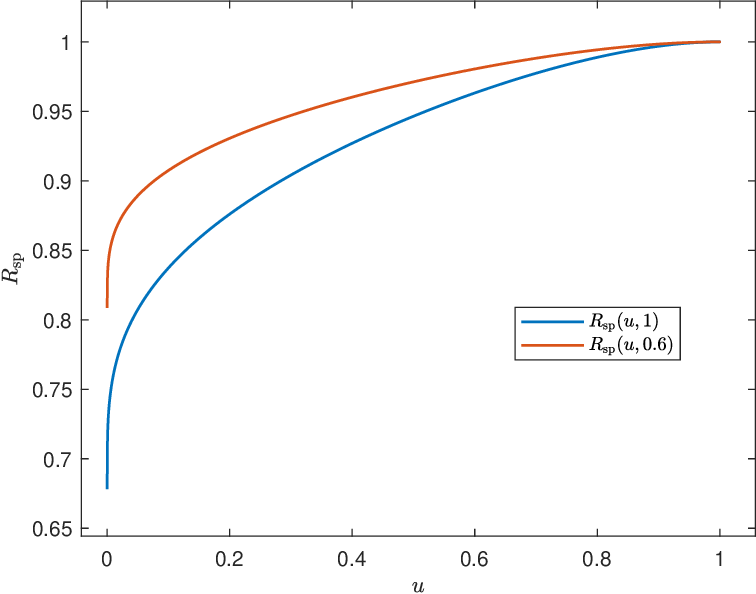}
    \includegraphics[width=0.4\linewidth]{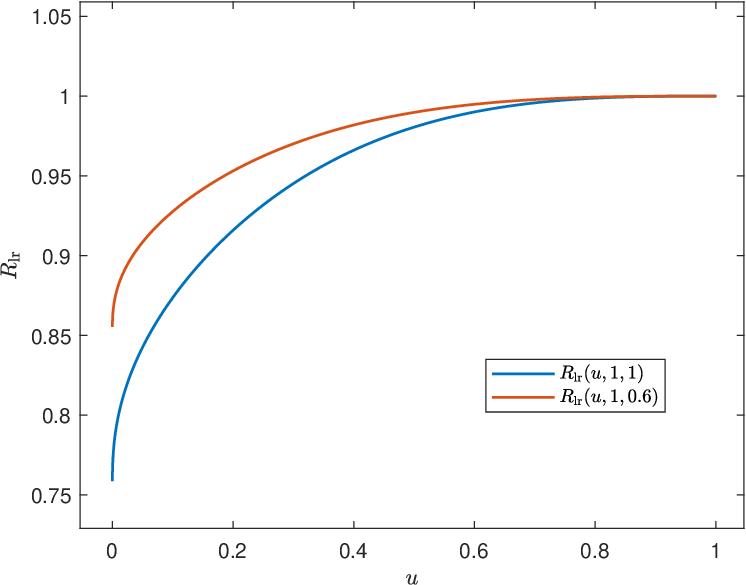}
    \caption{The left figure plots $R_{\rm sp}(u,1)$ and $R_{\rm sp}(u,0.6)$, showing    $\lim_{u\to 0^+}R_{\rm sp}(u,1)\approx 0.678$ and $\lim_{u\to 0^+}R_{\rm sp}(u,0.6)\approx 0.808$; the former indicates that for recovering $s$-sparse $\bx\in\mathbb{S}^{n-1}$ with nonzero entries being $\pm 1/\sqrt{s}$, if $s$ is fixed and $n\to \infty$, then PO-CS  requires no more than $0.68 \zeta_{\rm LN}(\bx;\|\cdot\|_1)$ phases to succeed, as we highlighted in the abstract. Similarly, the right figure plots $R_{\rm lr}(u,1,1)$ and $R_{\rm lr}(u,1,0.6)$, and we further report $\lim_{u\to 0^+} R_{\rm lr}(u,1,1)\approx 0.758$ and $\lim_{u\to 0^+} R_{\rm lr}(u,1,0.6)\approx 0.856$.}
    \label{fig:ratio}
\end{figure}

\section{Numerical Results}\label{sec4}
 In this section, we corroborate our phase transition threshold with numerical simulations. We focus on PO-CS of sparse vectors and low-rank matrices. Our results demonstrate that the empirical phase transition locations can be accurately predicted by the theoretical values, and that  $\zeta_{\rm PO}$ is inversely proportional to $\|\bx\|_1$ (or $\|\bX\|_{nu}$) under fixed $(n,s)$ (or $(p,q,r)$); see Figures \ref{fig:sparse}-\ref{fig:LR} and the associated texts. 
\begin{figure}[!ht]
    \centering
    \includegraphics[width=0.44\linewidth]{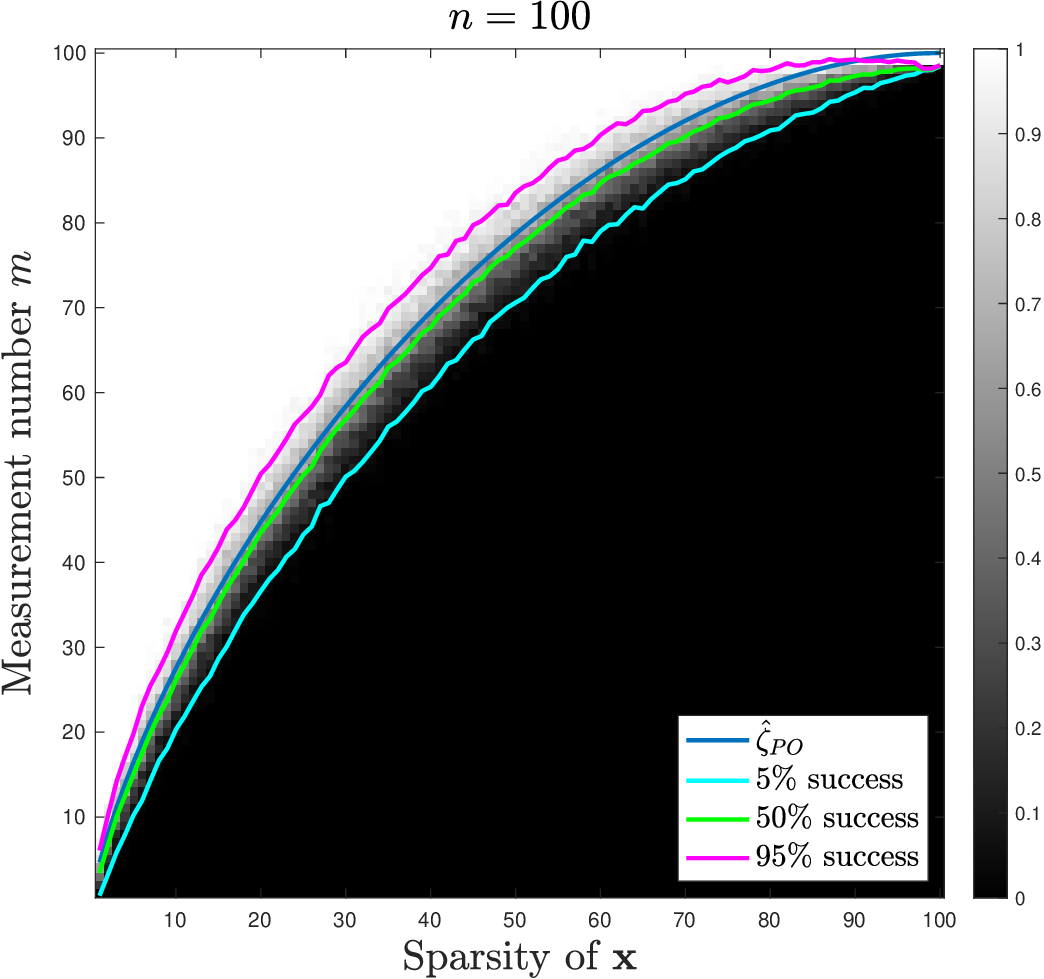}
    \includegraphics[width=0.45\linewidth]{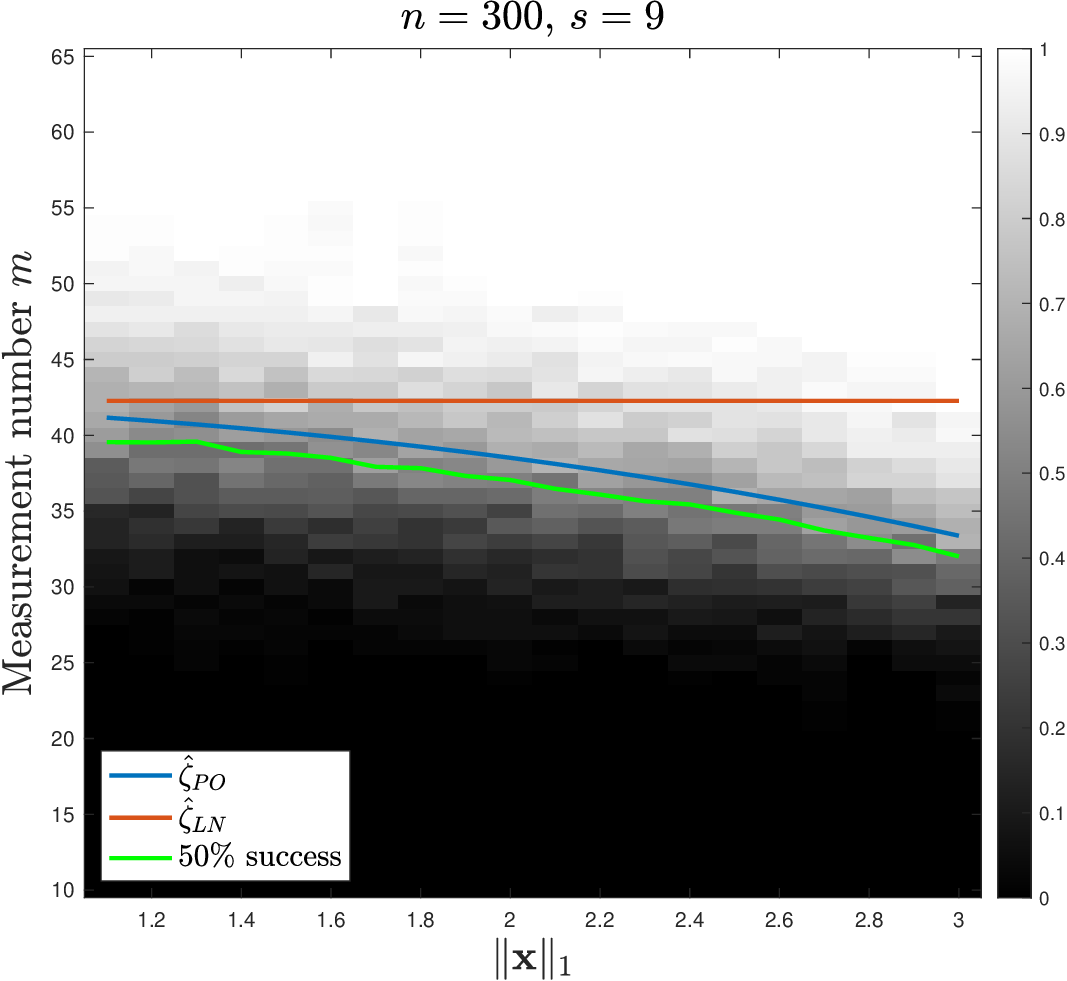}
    \caption{The left  figure shows that the empirical phase transitions of recovering equal amplitude sparse vectors in $\mathbb{R}^{100}$ are consistent with $\hat{\zeta}_{\rm PO}(\bx;\|\cdot\|_1)$.  The right figure shows the empirical success rates of recovering $9$-sparse signals in $\mathbb{R}^{300}$ with $\ell_1$-norm varying between $[1.1,3]$, confirming earlier phase transitions under larger $\|\bx\|_1$.}
    \label{fig:sparse}
\end{figure}

\begin{figure}[!ht]
    \centering
    \includegraphics[width=0.43\linewidth]{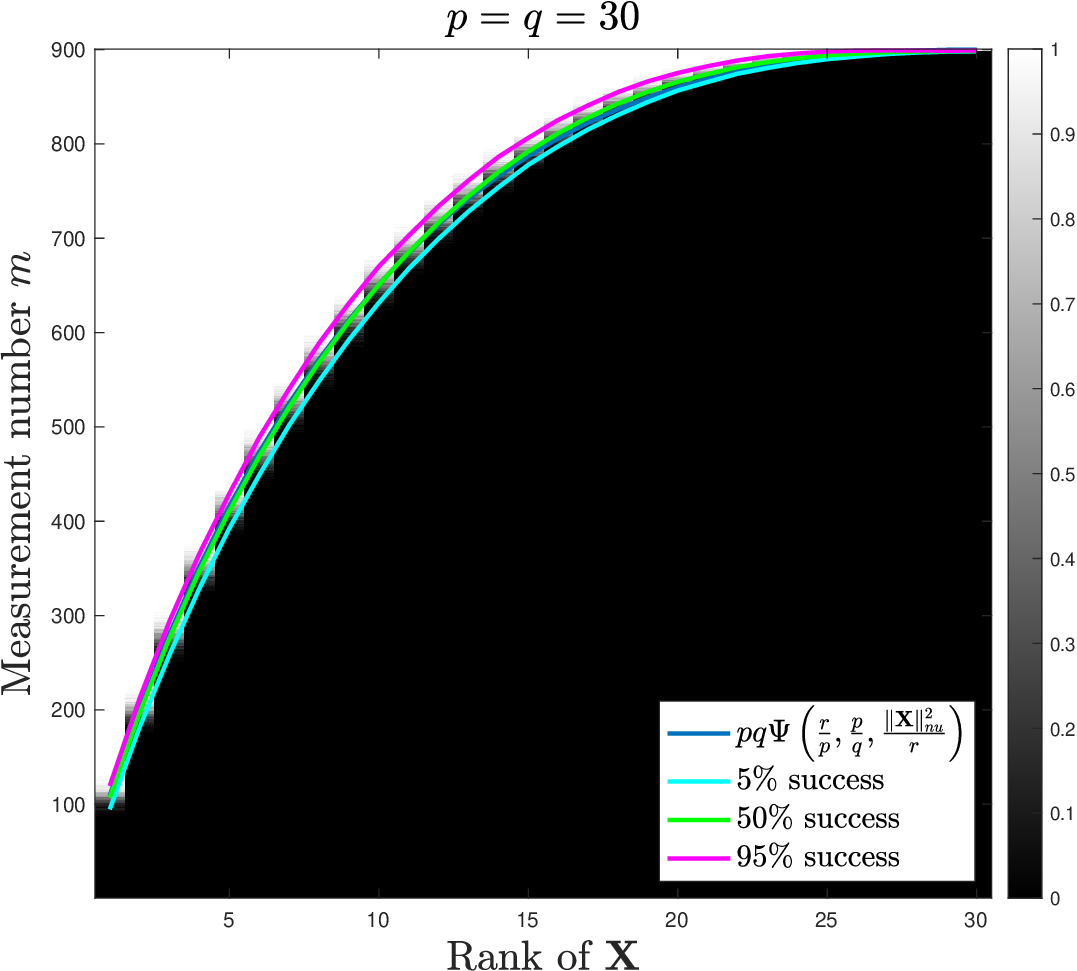}
    \includegraphics[width=0.425\linewidth]{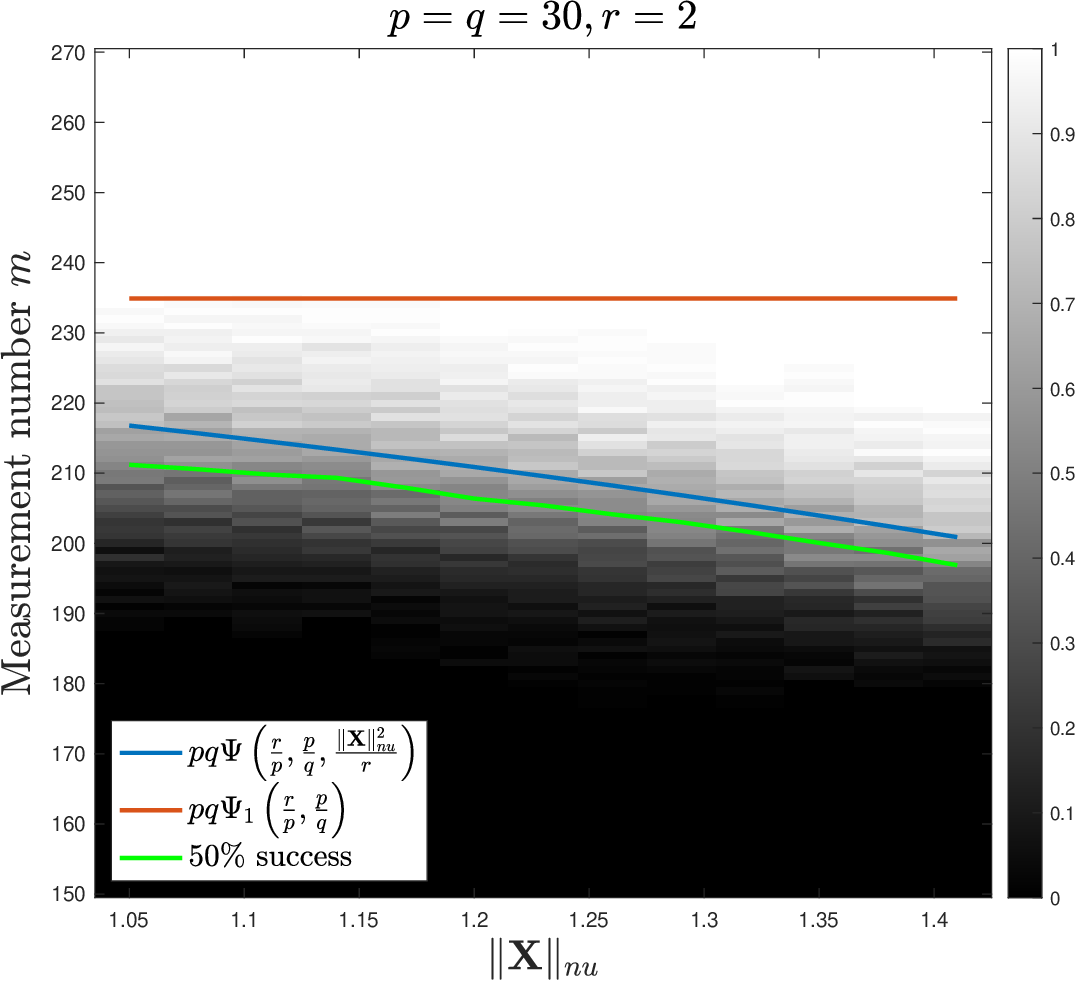}
    \caption{The left figure shows that the empirical phase transitions of recovering low-rank $30\times 30$ matrices can be precisely predicted by $pq\Psi(\frac{r}{p},\frac{p}{q},\frac{\|\bX\|_{nu}^2}{r})$. The right figure shows the empirical success rates of recovering rank-$2$ $30\times 30$ matrices with nuclear norm varying between $[1.05,\sqrt{2}]$, confirming earlier phase transitions under larger $\|\bX\|_{nu}$.}
    \label{fig:LR}
\end{figure}

  All experiments are performed in MATLAB using a MacBook Pro, where optimization problems are solved using the CVX package \cite{grant2014cvx}. We consider   a recovery  successful if the recovered signal is within $10^{-3}$ $\ell_2$-distance from the ground truth signal. For each setting, we simulate 100 independent trials and report the empirical success probability. For Figure \ref{fig:sparse}(left), we generate the underlying signals $\mathbf{x}\in \mathbb{S}^{n-1}$ with varying sparsities. For each sparsity level $s$, we consider equal amplitude signals that have non-zero entries $1/\sqrt{s}$. For Figure \ref{fig:sparse}(right), we keep the sparsity level $s$ constant and generate signals $\mathbf{x}$ with different $\ell_1$ norms. To do so, we consider signals whose non-zero entries take the same values except for one entry. We can then solve a quadratic equation to find the values of these entries, so that both the $\ell_1$ and $\ell_2$ norm requirements of $\mathbf{x}$ are satisfied. The same set of strategies is used to generate the singular values of the low-rank matrices $\mathbf{X}$ in Figure \ref{fig:LR}. The theoretical phase transitions in Figure \ref{fig:sparse} (resp. Figure \ref{fig:LR}) are calculated according to 
\eqref{61}--\eqref{62} (resp. \eqref{eq:zeta_PO}--\eqref{inftau}). The contours for specific success rates are calculated by fitting the data with logistic regression models. 
\section{Conclusion}\label{sec5}
 This paper establishes the phase transition curve for the phase-only compressed sensing (PO-CS) problem. It derives an exact formula for the phase transition location of a recovery algorithm for PO-CS which is based on linearizing the problem and applying basis pursuit. Furthermore, it demonstrates that the phase transition location can be approximated by the statistical dimension of the descent cone of a signal-dependent regularization, \(f_{\bx}(\bu) = f(\bu - (1 - \sqrt{2/\pi}\langle \bx, \bu\rangle)\bx)\), where $f$ is the norm used in the basis pursuit problem and $\bx$ is the true signal to be recovered. Using this result, the paper provides asymptotically exact formulas for the phase transition curve of PO-CS in recovering sparse signals and low-rank matrices. Comparisons with linear compressed sensing confirm that the phase transition for PO-CS occurs at lower values of \(m\).
 
\bibliographystyle{plain}
\bibliography{libr}

\vspace{3mm}

{\centering \huge \bf Appendix \par}

\begin{appendix}
    \section{Deferred Proofs}\label{appA}
    \subsection{Proof of Lemma \ref{lem:suffcon}}\label{provelem1}
    We should first emphasize that if \eqref{1.2} is successful, then $\hat{\bx} = \bx^\star$ where $\bx^{\star}$ is defined as $ \bx^\star =\frac{m\bx}{\|\bPhi\bx\|_1}$ . This is due to the fact that the solution has to satisfy $\bA_{\bz}\bu = \be_1$. 
     We first show that (\ref{succ_con}) implies $\bx^\sharp = \bx$. Note that it suffices to show $\hat{\bx}=\bx^\star$ under (\ref{succ_con}). To this end, we assume $\hat{\bx}\ne \bx^\star$ and show that this assumption will lead to a contradiction. Note that if $f(\hat{\bx})\le f(\bx^\star)$ and $\hat{\bx}\ne\bx^\star$, then 
     $$\bh := \frac{\hat{\bx}-\bx^\star}{\|\hat{\bx}-\bx^\star\|_2}\in \calT_f^*(\bx^\star)=\calT^*_f(\bx),$$
     where $\mathcal{T}_f^*(\bx^\star) = \calT_f^*(\bx)$ follows from the definition of descent cone and the assumption that $f(\cdot)$ is a norm. On the other hand, we have $\bA_{\bz}\hat{\bx}=\bA_{\bz}\bx^\star=\be_1$ and hence $\bA_{\bz}\bh=0$, which contradicts (\ref{succ_con}) that is evidently equivalent to $\min_{\bu\in \calT_f^*(\bx)}\|\bA_{\bz}\bu\|_2>0$. Therefore, we must have $\hat{\bx}=\bx^\star$.

     Next, we show (\ref{fail_con}) implies the failure of the recovery. 
 Recall that the dual cone of $\calC\in \mathbb{R}^n$ is defined as $\calC^\circ = \{\bw\in\mathbb{R}^n:\langle \bw, \bu\rangle\le 0,~\forall \bu \in \calC\}$.
     We begin with 
     \begin{align*}
          \min_{\bv\in\mathbb{S}^m}\min_{\bs\in(\calT_f(\bx))^\circ}\|\bs -\bA^\top_{\bz}\cdot\bv \|_2\ge \min_{\bv\in \mathbb{S}^m}\min_{\bs\in(\calT_f(\bx))^\circ}\max_{\bu\in \calT_f^*(\bx)}\big\langle \bu,\bA_{\bz}^\top\cdot\bv-\bs\big\rangle\ge \min_{\bv\in \mathbb{S}^m}\max_{\bu\in\calT_f^*(\bx)} \bv^\top\bA_{\bz}\bu >0,   
     \end{align*}
     where the second inequality is due to the definition of $(\calT_f(\bx))^\circ$. Since $\calT_f(\bx)$ is closed and not a subspace, according to \cite[Prop. 3.8]{oymak2018universality} we have $0\in \bA_{\bz}\calT_f^*(\bx)$, i.e., $\bA_{\bz}\bh=0$ for some $\bh\in\calT_f^*(\bx)$. This implies that $\bA_{\bz}(\bx^\star+t\bh)=\be_1$   for any $t\in\mathbb{R}$.
     By the definition of descent cones and $\calT_f^*(\bx)=\calT_f(\bx^\star)$, for some $t_0>0$ we have $f(\bx^\star+t_0\bh)\le f(\bx^\star)$. Therefore, $\bx^\star+t_0\bh$ is preferable over $\bx^\star$ to the program (\ref{1.2}), and thus $\hat{\bx}\ne \bx^\star$. In view of $\bx^\sharp = \hat{\bx}/\|\hat{\bx}\|_2$, we need to further show that $\hat{\bx}\ne \lambda\bx^\star$ for any $\lambda\in(0,1)\cup(1,\infty)$. In fact, $\hat{\bx}$ is feasible to (\ref{1.2}) and hence satisfies $\frac{1}{m}\Re(\bz^*\bPhi)\hat{\bx}=1$; if $\hat{\bx}=\lambda\bx^\star$ holds for some $\lambda>0$, then we have $\lambda\cdot \frac{1}{m}\Re(\bz^*\bPhi)\bx^\star =1$, which yields $\lambda =1$, contradicting $\hat{\bx}\ne\bx^\star$. Taken collectively, if (\ref{fail_con}) holds, then $\hat{\bx}\ne \lambda\bx^\star$ for any $\lambda>0$, and hence $\hat{\bx}\ne \lambda\bx$ for any $\lambda>0$. This leads to $\bx^\sharp\ne \bx$, as desired.  

     \subsection{Proof of Lemma \ref{lem2}}\label{provelem2}
     Define $\bPhi_{\bx}=\bPhi\bP_\bx^\top$ and note that since $P_x$ is an orthogonal matrix, $\bPhi_\bx$ has i.i.d. $\calN(0,1)+\calN(0,1)\bi$ entries. 
    Furthermore, denote the $j$-th column of $\bPhi_\bx$ by $\bphi_j$. Using these notations we can now express the phase-only observations as
    $$\bz = \sign(\bPhi\bx)=\sign(\bPhi \bP_\bx^\top \bP_\bx \bx) = \sign(\bPhi_\bx \be_1) = \sign(\bphi_1).$$
    Substituting $\bz$ in $\bA_{\bz}\bP_\bx^\top$ yields 
    \begin{align*}
        \bA_{\bz}\bP_{\bx}^\top& =\begin{bmatrix}
            \frac{1}{m}\Re(\bz^*\bPhi_\bx)\\
            \frac{1}{\sqrt{m}}\Im(\diag(\bz^*)\bPhi_\bx) 
        \end{bmatrix} = 
        \begin{bmatrix}
            \frac{1}{m}\Re(\sign(\bphi_1^*)[\bphi_1~\cdots~\bphi_n])\\
            \frac{1}{\sqrt{m}}\Im\big(\diag(\sign(\bphi_1^*))[\bphi_1~\cdots~\bphi_n]\big)
        \end{bmatrix}\\
        & = \begin{bmatrix}
            \frac{\|\bphi_1\|_1}{m} & \frac{\Re(\sign(\bphi_1^*)\bphi_2)}{m} & \cdots &  \frac{\Re(\sign(\bphi_1^*)\bphi_n)}{m} \\
            0 & \frac{\Im(\diag(\sign(\bphi_1^*))\bphi_2)}{\sqrt{m}} & \cdots & \frac{\Im(\diag(\sign(\bphi_1^*))\bphi_n)}{\sqrt{m}}
        \end{bmatrix} := \begin{bmatrix}
            L'\be_1 & \bA'
        \end{bmatrix},
    \end{align*}
    where to obtain the last equality we have defined $L'=\frac{\|\bphi_1\|_1}{m}$ and $\bA'\in \mathbb{R}^{(m+1)\times(n-1)}$ as the matrix constituted from the last $n-1$ columns of $\bA_{\bz}\bP_\bx^\top$. It is straightforward to see that $L'$ and $\bA'$ are independent of each other, and that  $L'$ has the same distribution as $L$ mentioned in the statement of Lemma \ref{lem2}. The only remaining part is to show that $\bA'\stackrel{d}{=}\frac{\bG}{\sqrt{m}}$. Since $\diag(\sign(\bphi_1^*))$ is a unitary matrix, we have $$\diag(\sign(\bphi_1^*))[\bphi_2,\cdots,\bphi_n]\sim \calN^{m\times (n-1)}(0,1)+ \calN^{m\times (n-1)}(0,1)\bi,$$
    which then implies that 
    $\Re\big(\diag(\sign(\bphi_1^*))[\bphi_2,\cdots,\bphi_n]\big)\sim \calN^{m\times (n-1)}(0,1)$ and $\Im\big(\diag(\sign(\bphi_1^*))[\bphi_2,\cdots,\bphi_n]\big)\sim \calN^{m\times (n-1)}(0,1)$ are independent. This further implies the following: (i) the last $m$ rows of $\bA'$ have i.i.d. $\frac{\calN(0,1)}{\sqrt{m}}$ entries; (2) the first row of $\bA'$ has independent entries distributed as $\frac{\calN(0,1)}{\sqrt{m}}$ (due to being the mean of $m$ i.i.d. $\calN(0,1)$ variables); (iii)  these two parts of $\bA'$ are independent. This proves that $\bA' \stackrel{d}{=}\bG/\sqrt{m}$. 
    \subsection{Proof of Theorem \ref{mainthm}}\label{provethm1}
    Given a fixed $\bx\in \mathbb{S}^{n-1}$, a norm $f$ and measurement number $m$, we denote the success probability as $\mathsf{P}_s(\bx;f,m) = \mathbbm{P}(\bx^\sharp=\bx)$, and the failure probability as $\mathsf{P}_f(\bx;f,m) = \mathbbm{P}(\bx^\sharp \ne \bx)$. Throughout the proof, we will use $\mathsf{P}_s$ and $\mathsf{P}_f$ for short.

    \paragraph{Lower bounding $\mathsf{P}_s:$} We first use Lemma \ref{lem:suffcon} to obtain a lower bound  for $\mathsf{P}_s$ in the following way: 
    \begin{align}\label{step11}
        \mathsf{P}_s &\ge \mathbbm{P}\Big(\min_{\bu\in\calT_f^*(\bx)}\max_{\bv\in\mathbb{S}^m}~\bv^\top \bA_{\bz}\bu>0\Big) = \mathbbm{P}\Big(\min_{\bu\in \bP_{\bx}\calT_f^*(\bx)}\max_{\bv\in\mathbb{S}^m}~\bv^\top \bA_{\bz}\bP_{\bx}^\top \bu >0\Big).
    \end{align}  
    Combining this equation  with Lemma \ref{lem2} leads to:  
    \begin{align}\label{13}
        &\mathsf{P}_s\ge  \mathbbm{P}\Big(\min_{\bu\in \bP_{\bx}\calT_f^*(\bx)}\max_{\bv\in\mathbb{S}^m}~\bv^\top \begin{bmatrix}
            L\be_1 & \frac{\bG}{\sqrt{m}}
        \end{bmatrix} \bu>0\Big)\\\label{expand}
        & = \mathbbm{P}\Big(\min_{\bu\in \bP_{\bx}\calT_f^*(\bx)}\max_{\bv\in\mathbb{S}^m}~Lu_1v_1 + \frac{\bv^\top\bG\tilde{\bu}}{\sqrt{m}}>0\Big)\\
        \label{introdelta}& \ge \mathbbm{P}\Big(\min_{\bu\in \bP_{\bx}\calT_f^*(\bx)}\max_{\bv\in\mathbb{S}^m}~\sqrt{m}Lu_1v_1 + \bv^\top\bG\tilde{\bu}\ge\delta\Big)
    \end{align}
    where $L$ and $\bG$ in (\ref{13}) are as described in Lemma \ref{lem2}. To obtian (\ref{expand}) we have defined  $\tilde{\bu}$ through $\bu= (u_1,\tilde{\bu}^\top)^\top$ (thus $\tilde{\bu}\in \mathbb{R}^{n-1}$) and denote the first entry of $\bv\in \mathbb{R}^{m+1}$ by $v_1$. Finally, (\ref{introdelta}) holds for arbitrary $\delta>0$. Now   introducing $\bg\in \mathbb{R}^{m+1}$ and $\bh\sim \mathbb{R}^{n-1}$ which have i.i.d. $\calN(0,1)$ entries, we  use (\ref{introdelta}) and Lemma \ref{gmt} to obtain 
    \begin{align}\label{16}
        & \mathsf{P}_s \ge 2 \mathbbm{P}\Big(\min_{\bu\in \bP_{\bx}\calT_f^*(\bx)}\max_{\bv\in\mathbb{S}^m}~\|\tilde{\bu}\|_2\bg^\top\bv +\|\bv\|_2 \bh^\top\tilde{\bu}+ \sqrt{m}Lu_1v_1>\delta\Big)-1\\\label{17}
        & = 2 \mathbbm{P}\Big(\min_{\bu\in \bP_{\bx}\calT_f^*(\bx)}~\Big[\bh^\top\tilde{\bu} + \max_{\bv\in\mathbb{S}^m}~\big[\|\tilde{\bu}\|_2\bg + \sqrt{m}Lu_1\be_1\big]^\top \bv\Big]>\delta\Big) -1\\ \label{18}
        & = 2\mathbbm{P}\Big(\min_{\bu\in \bP_{\bx}\calT_f^*(\bx)}~\bh^\top\tilde{\bu}+ \Big\|\|\tilde{\bu}\|_2\bg+ \sqrt{m}Lu_1\be_1\Big\|_2>\delta\Big) -1. 
    \end{align}
    Further using the symmetry of $\bh$   gives 
    \begin{align}\nn
       & \mathsf{P}_{s} \ge 2 \mathbbm{P}\Big(\max_{\bu\in\bP_{\bx}\calT_f^*(\bx)}~\bh^\top\tilde{\bu}-\Big\|\|\tilde{\bu}\|_2 \bg + \sqrt{m}Lu_1\be_1\Big\|_2<-\delta\Big)-1\\
        & = 2 \mathbbm{P}\Big(\forall \bu\in\bP_{\bx}\calT_f^*(\bx),~\bh^\top\tilde{\bu}<\Big\|\|\tilde{\bu}\|_2 \bg + \sqrt{m}Lu_1\be_1\Big\|_2-\delta\Big)-1\label{lowerps}
    \end{align}
    that holds for arbitrary $\delta>0$. 
    \paragraph{Lower bounding $\mathsf{P}_f:$}  
    We can similarly find the following lower bound for $\mathsf{P}_f$:
    \begin{align}\label{step1}
         & \mathsf{P}_f \ge \mathbbm{P}\Big(\min_{\bv\in \mathbb{S}^m}\max_{\bu\in \bP_{\bx}\calT_f^*(\bx)}~\bv^\top \bA_{\bz}\bP_{\bx}^\top\bu >0 \Big) \\\label{step2}
        & \ge \mathbbm{P}\Big(\min_{\bv\in\mathbb{S}^m}\max_{\bu\in \bP_{\bx}\calT_f^*(\bx)}~\sqrt{m}Lu_1v_1 + \bv^\top\bG\tilde{\bu}\ge\delta\Big) \\\label{step3}
        & \ge 2 \mathbbm{P}\Big(\min_{\bv\in\mathbb{S}^m}\max_{\bu\in \bP_{\bx}\calT_f^*(\bx)}~\|\tilde{\bu}\|_2\bg^\top\bv +\|\bv\|_2 \bh^\top\tilde{\bu}+ \sqrt{m}Lu_1v_1>\delta\Big)-1 \\\label{step4}
        &\ge 2\mathbbm{P}\Big(\max_{\bu\in \bP_{\bx}\calT_f^*(\bx)}\min_{\bv\in\mathbb{S}^m}~\|\tilde{\bu}\|_2\bg^\top\bv + \bh^\top\tilde{\bu}+ \sqrt{m}Lu_1v_1>\delta\Big)-1 \\\label{step5}
        & = 2\mathbbm{P}\Big(\max_{\bu\in \bP_{\bx}\calT_f^*(\bx)}~\bh^\top\tilde{\bu}-\Big\|\|\tilde{\bu}\|_2\bg +\sqrt{m}Lu_1\be_1\Big\|_2>\delta\Big)-1\\
        & = 2\mathbbm{P}\Big(\exists \bu\in \bP_{\bx}\calT_f^*(\bx),~\bh^\top\tilde{\bu}>\Big\|\|\tilde{\bu}\|_2\bg +\sqrt{m}Lu_1\be_1\Big\|_2+\delta\Big)-1\label{lowerpf}
    \end{align}
    for arbitrary $\delta>0$; note that (\ref{step1}) corresponds to (\ref{step11}), (\ref{step2}) corresponds to (\ref{13})--(\ref{introdelta}),  and (\ref{step3}) corresponds to (\ref{16}). Furthermore, (\ref{step4}) holds because changing $\min_\bv\max_{\bu}$ to $\max_{\bu}\min_\bv$ cannot increase the final value, and in (\ref{step5}) we optimize over $\bv\in\mathbb{S}^m$ as in (\ref{17})--(\ref{18}). 

    \paragraph{Concentration of $\big\|\|\tilde{\bu}\|_2\bg+\sqrt{m}Lu_1\be_1\big\|_2$:} Next, we establish the concentration bound for $\|\|\tilde{\bu}\|_2\bg+\sqrt{m}Lu_1\be_1\|_2$. First note that 
    \begin{align}
        \label{squarenorm}\big\|\|\tilde{\bu}\|_2\bg+\sqrt{m}Lu_1\be_1\big\|_2^2 = \|\tilde{\bu}\|_2^2 \|\bg\|_2^2 + mL^2u_1^2 + 2 \sqrt{m}Lu_1\|\tilde{\bu}\|_2g_1  
        \end{align}
    where $g_1$ denotes the first entry of $\bg\sim \calN(0,\bI_{m+1})$. Note that $\|\bg\|_2^2$ follows the Chi-squared distribution with $m+1$ degrees of freedom, and thus the concentration of (e.g., \cite{laurent2000adaptive}) Chi-squared distribution implies the following for any $t\in(0,1)$:
    \begin{align}\label{concen1}
        \mathbbm{P}\Big(\Big|\|\bg\|_2^2-(m+1)\Big| \le 5mt\Big) \ge 1-2\exp(-mt^2). 
    \end{align}
    Recall from Lemma \ref{lem2} that $L\stackrel{d}{=}\frac{1}{m}\sum_{i=1}^m ({p_i^2+q_i^2})^{1/2}$ where $p_i,q_i$ are i.i.d. $\calN(0,1)$ random variables. Furthermore, we note that $\frac{1}{m}\sum_{i=1}^m({p_i^2+q_i^2})^{1/2}$ is a $\frac{1}{\sqrt{m}}$-Lipschitz function of $(p_i,q_i)_{i=1}^m$ since 
    $$\Big|\frac{1}{m}\sum_{i=1}^m\sqrt{p_i^2+q_i^2}-\frac{1}{m}\sum_{i=1}^m\sqrt{\tilde{p}_i^2+\tilde{q}_i^2}\Big|\le \frac{1}{m}\sum_{i=1}^m \sqrt{(p_i-\tilde{p}_i)^2+(q_i-\tilde{q}_i)^2}\le \frac{1}{\sqrt{m}}\Big(\sum_{i=1}^m \big[(p_i-\tilde{p}_i)^2+(q_i-\tilde{q}_i)^2\big]\Big)^{1/2}.$$
    Thus, by the concentration of Lipschitz functions of Gaussian vectors (e.g., see \cite[Eq. (1.6)]{ledoux2013probability}),  $L$ concentrates around $\mathbbm{E}L=\sqrt{\frac{\pi}{2}}$  and obeys 
    \begin{align}\label{concen2}
        \mathbbm{P}\Big(\Big|L-\sqrt{\frac{\pi}{2}}\Big|\le\sqrt{2}t\Big)\ge 1-2\exp(-mt^2)
    \end{align}
    for any $t>0$. By restricting to $t\in(0,\frac{1}{4}]$, the above event implies $|L-\sqrt{\frac{\pi}{2}}|\le \frac{\sqrt{2}}{4}$ and hence $L\le 2$. On this event, we have $|2\sqrt{m}Lu_1\|\tilde{\bu}\|_2g_1|\le 2\sqrt{m}|g_1|$ (by using $|2u_1\|\tilde{\bu}\|_2|\le u_1^2+\|\tilde{\bu}\|_2^2=1$), which together with the standard tail bound  $\mathbbm{P}(|g_1|\ge t)\le \exp(-\frac{t^2}{2})$ yields the following: 
    \begin{align}
        \label{concen3}\mathbbm{P}\Big(\big|2\sqrt{m}Lu_1\|\tilde{\bu}\|_1g_1\big|\le 3mt\Big)\ge 1-\exp(-mt^2).  
    \end{align}
    Considering all the results mentioned above, we conclude that for any $t\in(0,\frac{1}{4}]$, with probability at least $1-5\exp(-mt^2)$ the concentration bounds in (\ref{concen1})--(\ref{concen3}) hold. Substituting these bounds into (\ref{squarenorm}) and performing some algebra, we reach the following bound: for any $t\in [\frac{1}{m},\frac{1}{4}]$, it holds with probability at least $1-5\exp(-mt^2)$ that \begin{align*}
       m\Big(\|\tilde{\bu}\|_2^2+\frac{\pi u_1^2}{2}\Big) -8mt \le \big\|\|\tilde{\bu}\|_2\bg+\sqrt{m}Lu_1\be_1\big\|_2^2\le m\Big(\|\tilde{\bu}\|_2^2+\frac{\pi u_1^2}{2}\Big) + 9mt.
    \end{align*}
    This further yields that if $t\in [\frac{1}{m},\frac{4}{81}]$, with probability at least $1-5\exp(-mt^2)$ the event
    \begin{align*}
    \mathscr{E}=\left\{\sqrt{m} \Big[\|\tilde{\bu}\|_2^2+\frac{\pi u_1^2}{2}\Big]^{1/2} - \frac{9\sqrt{m}t}{2}  \le \big\|\|\tilde{\bu}\|_2\bg+\sqrt{m}Lu_1\be_1\big\|_2 \le \sqrt{m} \Big[\|\tilde{\bu}\|_2^2+\frac{\pi u_1^2}{2}\Big]^{1/2} +\frac{9\sqrt{m}t}{2}\right\}
    \end{align*}
    holds. 

    \paragraph{$\mathsf{P}_s\to 1$ when $m$ is above $\zeta_{\rm PO}(\bx;f)$:} We are now ready to show $\mathsf{P}_s\to 1$ when $m$ exceeds $\zeta_{\rm PO}(\bx;f)$ by further lower bounding $\mathsf{P}_s$ via the event $\mathscr{E}$.  If we define $\mathscr{E}_s:=\{\forall\bu\in \bP_{\bx}\calT_f^*(\bx),~\bh^\top\tilde{\bu}<\big\|\|\tilde{\bu}\|_2\bg +\sqrt{m}Lu_1\be_1\big\|_2-\delta\}$,  then from (\ref{lowerps}) we have $\mathsf{P}_s \ge 2\mathbbm{P}(\mathscr{E}_s)-1$. This implies that 
    \begin{align}\nn
        &\mathsf{P}_s \ge 2 \mathbbm{P}(\mathscr{E}_s\text{ and }\mathscr{E})-1\\\label{33}
        &\ge 2 \mathbbm{P}\Big(\Big\{\forall\bu\in\bP_{\bx}\calT_f^*(\bx),~\bh^\top\tilde{\bu}\le \sqrt{m} \Big[\|\tilde{\bu}\|_2^2+\frac{\pi u_1^2}{2}\Big]^{1/2} - 5\sqrt{m}t\Big\}\text{ and }\mathscr{E}\Big)-1 \\
        & \ge 2 \mathbbm{P}\left(\forall\bu\in\bP_{\bx}\calT_f^*(\bx),~\frac{\bh^\top\tilde{\bu}}{(\|\tilde{\bu}\|_2^2+\frac{\pi u_1^2}{2})^{1/2}}\le \sqrt{m} (1-5t)\right) -2 \mathbbm{P}(\mathscr{E}^c) - 1\\
      \label{35}  &\ge 2 \mathbbm{P}\left(\sup_{\bu\in\bP_{\bx}\calT_f^*(\bx)}\frac{\bh^\top\tilde{\bu}}{(\|\tilde{\bu}\|_2^2+\frac{\pi u_1^2}{2})^{1/2}}\le \sqrt{m} (1-5t)\right) -1 - 10\exp(-mt^2).
    \end{align}
    To obtain (\ref{33}) we have used $\mathscr{E}$ to make the event $\mathscr{E}_s$ more restrictive and set $\delta=\frac{\sqrt{m}t}{4}$ (recall that $\delta$ can be arbitrary positive number).
    To obtain a lower bound for \eqref{35} we prove a concentration bound for $$F(\bh):=\sup_{\bu\in \bP_{\bx}\calT_f^*(\bx)}\frac{\bh^\top\tilde{\bu}}{(\|\tilde{\bu}\|_2^2+\pi u_1^2/2)^{1/2}}.$$
    Note that $F(\bh)$ is a $1$-Lipschitz function of the Gaussian rancom vector $\bh \sim\calN(0,\bI_{n-1})$ because 
    $$|F(\bh)-F(\bh_1)| \le \sup_{\bu\in \bP_{\bx}\calT_f^*(\bx)}\frac{|(\bh-\bh_1)^\top\tilde{\bu}|}{(\|\tilde{\bu}\|_2^2+\pi u_1^2/2)^{1/2}}\le \|\bh-\bh_1\|_2.$$
    Thus, the concentration of Lipschitz functions of Gaussian vectors (e.g., see \cite[Eq. (1.6)]{ledoux2013probability}) proves that
    \begin{align}\label{concen_Fh}
        \mathbbm{P}\Big(\big|F(\bh)-\mathbbm{E}[F(\bh)]\big|\le  \sqrt{2m}t\Big)\ge 1- 2\exp(-mt^2).
    \end{align}
    Suppose $m\ge \frac{(\mathbbm{E}[F(\bh)])^2}{(1-7t)^2}$. Then we have $\sqrt{m}(1-5t)\ge \mathbbm{E}[F(\bh)]+2\sqrt{m}t>\mathbbm{E}[F(\bh)]+\sqrt{2m}t$, which allows us to continue from (\ref{35}) to derive 
    \begin{align*}
       &\mathsf{P}_s \ge 2 \mathbbm{P}\big(F(\bh)\le \sqrt{m}(1-5t)\big) - 1- 10\exp(-mt^2) \\
       & \ge  2 \mathbbm{P}\big(F(\bh)\le \mathbbm{E}[F(\bh)]+\sqrt{2m}t\big) - 1- 10\exp(-mt^2) \ge 1-14\exp(-mt^2).
    \end{align*}

    \paragraph{$\mathsf{P}_f\to 1$ when $m$ is below $\zeta_{\rm PO}(\bx;f)$:} Define $\mathscr{E}_f := \{\exists \bu \in\bP_{\bx}\calT_f^*(\bx),~\bh^\top\tilde{\bu}>\|\|\tilde{\bu}\|_2\bg+\sqrt{m}Lu_1\be_1\|_2+\delta\}$. Then (\ref{lowerpf}) implies that $\mathsf{P}_f \ge 2\mathbbm{P}(\mathscr{E}_f)-1$. By using $\mathscr{E}$ to make $\mathscr{E}_f$ more restrictive and setting $\delta = \frac{\sqrt{m}t}{4}$, we obtain 
    \begin{align*}
        &\mathsf{P}_f \ge 2\mathbbm{P}(\mathscr{E}_f\text{ and }\mathscr{E})-1 \\
        &\ge 2\mathbbm{P}\Big(\Big\{\exists \bu\in\bP_{\bx}\calT_f^*(\bx),~\bh^\top\tilde{\bu}\ge \sqrt{m}\Big[\|\tilde{\bu}\|_2^2+\frac{\pi u_1^2}{2}\Big]^{1/2}+5\sqrt{m}t \Big\}\text{ and }\mathscr{E}\Big)-1 \\
        &\ge 2\mathbbm{P}\Big(\exists\bu\in\bP_{\bx}\calT_f^*(\bx),~ \frac{\bh^\top\tilde{\bu}}{(\|\tilde{\bu}\|_2^2+\pi u_1^2/2)^{1/2}} \ge \sqrt{m}(1+5t)\Big)-2 \mathbbm{P}(\mathscr{E}^c) -1 \\
        &\ge 2\mathbbm{P}\big(F(\bh)\ge \sqrt{m}(1+5t)\big)-1-10\exp(-mt^2).
        \end{align*}
        If $m\le \frac{(\mathbbm{E}[F(\bh)])^2}{(1+7t)^2}$, then we have $\mathbbm{E}[F(\bh)]\ge\sqrt{m}(1+7t)$ and hence $\sqrt{m}(1+5t)\le \mathbbm{E}[F(\bh)]-2\sqrt{m}t$. Therefore, we can use (\ref{concen_Fh}) to reach
        \begin{align*}
            \mathsf{P}_f \ge 2\mathbbm{P}\big(F(\bh)\ge \mathbbm{E}[F(\bh)]-2\sqrt{m}t\big)-1-10\exp(-mt^2)\ge 1-14\exp(-mt^2).
        \end{align*}

        \paragraph{Summary:}  Therefore, for any $t\in[\frac{1}{m},\frac{1}{4}]$, $m\ge \frac{(\mathbbm{E}[F(\bh)])^2}{(1-7t)^2}$ implies $\mathsf{P}_s \ge 1-14\exp(-mt^2)$, and $m \le \frac{(\mathbbm{E}[F(\bh)])^2}{(1+7t)^2}$ implies $\mathsf{P}_f \ge 1-14\exp(-mt^2)$. %Replacing $7t$ by $w$, we obtain the following: for any $w \in [\frac{7}{m},\frac{7}{4}]$, we have that $m\ge \frac{(\mathbbm{E}[F(\bh)])^2}{(1-w)^2}$ implies $\mathsf{P}_s \ge 1-14\exp(-\frac{mw^2}{49})$, and that $m\le \frac{(\mathbbm{E}[F(\bh)])^2}{(1+w)^2}$ implies $\mathsf{P}_f \ge 1-14\exp(-\frac{mw^2}{49})$. Regarding the success probability, we restrict our attention to $w\in [\frac{7}{m},\frac{1}{2}]$ and have   $\frac{1}{(1-w)^2}\le (1+2w)^2\le 1+6w$, and so $m \ge (1+6w)(\mathbbm{E}[F(\bh)])^2$ implies $\mathsf{P}_s \ge 1-14\exp(-\frac{mw^2}{49})$. 
        Note that when $t$ is small enough, we have $\frac{1}{(1-7t)^2}\le (1+8t)^2\le 1+17t$, $\frac{1}{(1+7t)^2}\ge (1-7t)^2\ge 1-17t$. Thus, replacing $17t$ by $w$ yields the following for a small enough $c_1$: for any $w\in [\frac{17}{m},c_1]$, $m\ge (1+w)(\mathbbm{E}[F(\bh)])^2$ implies $\mathsf{P}_s\ge 1-14\exp(-\frac{mt^2}{289})$, and $m\le (1-w)(\mathbbm{E}[F(\bh)])^2$ implies $\mathsf{P}_f \ge 1-14\exp(-\frac{mt^2}{289})$. All that remains is to show $\zeta_{\rm PO}(\bx;f) = (\mathbbm{E}[F(\bh)])^2$, that is     
        \begin{align}\label{desir}
        \mathbbm{E}\sup_{\bu\in\bP_{\bx}\calT_f^*(\bx)}\frac{\bh^\top \tilde{\bu}}{(\|\tilde{\bu}\|_2^2+\frac{\pi u_1^2}{2})^{1/2}} = \mathbbm{E}\sup_{{\bu\in\calT_f(\bx),\|\bQ_{\bx}\bu\|_2=1}}\big\langle(\bI_n-\bx\bx^\top)\bg,\bu\big\rangle    
        \end{align}
        where $\bh\in \calN(0,\bI_{n-1}),~\bg\sim\calN(0,\bI_n)$. 
        This is achieved by the following arguments:
        \begin{align}\nn
        &\mathbbm{E}\sup_{\bu\in\bP_{\bx}\calT_f^*(\bx)}\frac{\bh^\top \tilde{\bu}}{(\|\tilde{\bu}\|_2^2+\frac{\pi u_1^2}{2})^{1/2}} = \mathbbm{E}\sup_{\bu\in\bP_{\bx}\calT_f^*(\bx)}\frac{\langle (\bI_n-\be_1\be_1^\top)\bg,\bu\rangle}{(\|\bu\|_2^2+(\frac{\pi}{2}-1)(\bu^\top\be_1)^2)^{1/2}}\\\label{46}
            & = \mathbbm{E}\sup_{\bu\in\calT^*_f(\bx)} \frac{\langle \bP_{\bx}^\top(\bI_n-\be_1\be_1^\top)\bP_\bx\bg,\bu\rangle}{(\|\bu\|_2^2+(\frac{\pi}{2}-1)(\bu^\top\bP_\bx^\top\be_1)^2)^{1/2}} =  \mathbbm{E}\sup_{\bu\in\calT^*_f(\bx)} \frac{\langle (\bI_n-\bx\bx^\top)\bg,\bu\rangle}{\|\bQ_{\bx}\bu\|_2} \\& \label{47}= \mathbbm{E}\sup_{\bu\in\calT_f(\bx)} \frac{\langle (\bI_n-\bx\bx^\top)\bg,\bu\rangle}{\|\bQ_{\bx}\bu\|_2} =  \mathbbm{E}\sup_{\substack{\bu\in\calT_f(\bx)\\\|\bQ_\bx\bu\|_2=1}}\langle(\bI_n-\bx\bx^\top)\bg,\bu\rangle.
        \end{align}
        To obtain (\ref{46}) we substitute $\bu$ with $\bP_{\bx}\bu$ and use $\bP_\bx\bg\stackrel{d}{=}\bg$ in the first equality. Recall that $\bP_{\bx}^\top\be_1=\bx$ and observe $\|\bu\|_2^2+(\frac{\pi}{2}-1) (\bu^\top\bx)^2= \|\bQ_{\bx}\bu\|_2^2$ in the second equality; then, we remove the constraint $\bu\in\mathbb{S}^{n-1}$ in the first equality of (\ref{47}) because $\frac{\langle(\bI_n-\bx\bx^\top)\bg,\bu\rangle}{\|\bQ_{\bx}\bu\|_2}$ is homogeneous in $\bu$, and the final equality holds similarly due to the homogeneity. This completes the proof. 

        \subsection{Proof of Proposition \ref{pro1}}\label{provepro1}
         Recall that $\bQ_{\bx} = \bI_n+ (\sqrt{\frac{\pi}{2}}-1)\bx\bx^\top$ and $\bQ_{\bx}^{-1} = \bI_n- (1-\sqrt{\frac{2}{\pi}})\bx\bx^\top$. If we define $\bw= \bQ_{\bx}\bu$, then we have 
    \begin{align*}
             & \mathbbm{E}\sup_{\substack{\bu\in \calT_f(\bx)\\ \|\bQ_{\bx}\bu\|_2= 1}}\big\langle (\bI_n-\bx\bx^\top)\bg,\bu\big\rangle = \mathbbm{E}\sup_{\substack{\bw\in\bQ_{\bx}\calT_f(\bx)\\\|\bw\|_2=1}} \big\langle  \bQ_{\bx}^{-1}(\bI_n-\bx\bx^\top)\bg,\bw\big\rangle \\
             & =  \mathbbm{E}\sup_{\substack{\bw\in\bQ_{\bx}\calT_f(\bx)\\\|\bw\|_2=1}} \big\langle  (\bI_n-\bx\bx^\top)\bg,\bw\big\rangle  \stackrel{(i)}{=} \mathbbm{E} \sup_{\bw\in \calT^*_{f_{\bx}}(\bx)} \big\langle(\bI_n-\bx\bx^\top)\bg,\bw\big\rangle.
    \end{align*}
   Note that in the above equations we have used the notation $\calT_{f_{\bx}}^*(\bx) = \calT_{f_{\bx}}(\bx)\cap \mathbb{S}^{n-1}$. Furthermore, $(i)$ holds because $\bQ_{\bx}\bx=\sqrt{\frac{\pi}{2}}\bx$ and hence
    \begin{align*}
        & \bQ_{\bx}\calT_f(\bx) = \big\{\bQ_{\bx}\bv: \exists t>0,~f(\bx+t\bv)\le f(\bx)\big\} \\
        & = \big\{\bw: \exists t>0,~f\big(\bQ_{\bx}^{-1}\big[\bQ_{\bx}\bx+ t\bw\big]\big)\le f\big(\bQ_{\bx}^{-1}\bQ_{\bx}\bx\big)\big\}\\
        & = \big\{\bw : \exists t>0,~f(\bQ_{\bx}^{-1}[\bx+t\bw])\le f(\bQ_{\bx}^{-1}\bx)\big\} = \calT_{f_{\bx}}(\bx).
    \end{align*}
    Moreover, we have $\mathbbm{E}\sup_{\bw\in \calT^*_{f_{\bx}}(\bx)}\langle (\bI_n-\bx\bx^\top)\bg,\bw\rangle\le \omega\big(\calT^*_{f_{\bx}}(\bx)\big)$ and $\mathbbm{E}\sup_{\bw\in \calT^*_{f_{\bx}}(\bx)}\langle (\bI_n-\bx\bx^\top)\bg,\bw\rangle\ge \omega\big(\calT^*_{f_{\bx}}(\bx)\big) - \mathbbm{E}\sup_{\bw\in \calT^*_{f_{\bx}}(\bx)}|\bx^\top\bg||\bx^\top\bw|\ge \omega\big(\calT^*_{f_{\bx}}(\bx)\big)-\sqrt{2/\pi}$. Taking square and using $\omega^2(\calC^*)\le \delta(\calC)\le \omega^2(\calC^*)+1$ yields the desired bound. 
        \subsection{Proof of Theorem \ref{thm2}}\label{provethm2}
     We only need to estimate $\delta(\calT_{f_{\bx}}(\bx))$. By using \cite[Thm. 4.3]{amelunxen2014living} we obtain 
    \begin{align}\label{livingeq}
    \inf_{\tau\ge 0}~\mathbbm{E}\Big[\dist^2(\bg, \tau\cdot\partial f_{\bx}(\bx))\Big] - \frac{2\rad(\partial f_{\bx}(\bx))}{f_{\bx}(\bx)}\le  \delta(\calT_{f_{\bx}}(\bx))\le \inf_{\tau\ge 0}~\mathbbm{E}\Big[\dist^2(\bg, \tau\cdot \partial f_{\bx}(\bx))\Big].
    \end{align}
    Moreover, we have $f_{\bx}(\bx) = f(\bQ_{\bx}^{-1}\bx) = \sqrt{2/\pi}f(\bx)$ and
    \begin{align*}
        &\partial f_{\bx}(\bx) = \big\{\bu\in\mathbb{R}^n:f_{\bx}(\by)\ge f_{\bx}(\bx) + \langle \bu, \by-\bx\rangle,~\forall\by\in\mathbb{R}^n\big\}\\
        &= \big\{\bu\in\mathbb{R}^n:f(\bQ_{\bx}^{-1}\by)\ge \sqrt{\frac{2}{\pi}}f(\bx)+ \langle \bu, \by-\bx\rangle ,~\forall \by\in \mathbb{R}^n\big\}\\
        &= \Big\{\bu\in \mathbb{R}^n:\sqrt{\frac{2}{\pi}}f(\bw)\ge \sqrt{\frac{2}{\pi}}f(\bx)+ \big\langle \bu, \sqrt{\frac{2}{\pi}}\bQ_{\bx}\bw -\sqrt{\frac{2}{\pi}}\bQ_{\bx}\bx\big\rangle,~\forall\bw\in\mathbb{R}^n\Big\}\\
        & = \{\bQ_{\bx}^{-1}\bv :f(\bw)\ge f(\bx)+ \langle \bv,\bw-\bx\rangle\} = \bQ_{\bx}^{-1}\partial f(\bx),
    \end{align*}
    where in the third equality we substitute $\by$ by $\sqrt{2/\pi}\bQ_{\bx}\bw$. Substituting these into (\ref{livingeq}) yields 
    $$\hat{\zeta}_{\rm PO}(\bx;f) - \frac{\sqrt{2\pi}}{f(\bx)} \rad(\bQ_{\bx}^{-1}\partial f(\bx)) \le \delta(\calT_{f_{\bx}}(\bx))\le \hat{\zeta}_{\rm PO}(\bx;f).$$
    Finally,  combining this result with Proposition \ref{pro1} establishes the desired result. 
    \subsection{Proof of Theorem \ref{spa_formu}}\label{provethm3}
     We only need to prove (\ref{61})–(\ref{62}), as (\ref{63}) directly follows from Theorem \ref{thm2}. Note that we have $$\partial f(\bx) = \big\{\bu=(u_i)\in\mathbb{R}^n:u_i=\sign(x_i),~i\in {\rm supp}(\bx); ~|u_i|\le1 ,~i\notin {\rm supp}(\bx)\big\},$$ 
     which leads to $$\bQ_{\bx}^{-1}\partial f(\bx) = \big\{\bu=(u_i)\in\mathbb{R}^n:u_i = \sign(x_i)-(1-\sqrt{2/\pi})\|\bx\|_1x_i,~i\in{\rm supp}(\bx);~|u_i|\le 1,~i\notin {\rm supp}(\bx)\big\}.$$ Therefore, we have the calculations
    \begin{align*}
        &\mathbbm{E}\Big[\dist^2(\bg,\tau\cdot \bQ_{\bx}^{-1}\partial f(\bx))\Big]\\&= \mathbbm{E}\Big[\sum_{i\in{\rm supp}(\bx)}\Big(g_i-\tau\sign(x_i)+\tau\big(1-\sqrt{2/\pi}\big)\|\bx\|_1x_i\Big)^2+\sum_{i\notin{\rm supp}(\bx)}\big({\rm shrink}(g_i;\tau)\big)^2\Big]\\
        & = \sum_{i\in{\rm supp}(\bx)}\Big(1+\tau^2+\tau^2(1-\sqrt{2/\pi})^2\|\bx\|_1^2x_i^2-2\tau^2(1-\sqrt{2/\pi})|x_i|\|\bx\|_1\Big) + (n-s)\mathbbm{E}\Big[\big({\rm shrink}(g_i;\tau)\big)^2\Big]\nn\\
        & = s(1+\tau^2)- \tau^2 \|\bx\|_1^2 \big(1-\frac{2}{\pi}\big)+ (n-s)\sqrt{\frac{2}{\pi}}\int_{\tau}^\infty (w-\tau)^2\exp\big(-\frac{w^2}{2}\big)~\text{d}w,
    \end{align*}
    where ${\rm shrink}(g_i;\tau)$ is the thresholded version of $g_i$ that equals $0$ if $|g_i|\le \tau$, and equals $g_i-\tau$ if $g_i>\tau$, and equals $g_i+\tau$ if $g_i<-\tau$. It is straightforward to see that $\mathbbm{E}[{\rm shrink}^2(g_i;\tau)] = \sqrt{\frac{2}{\pi}}\int_{\tau}^\infty (w-\tau)^2\exp(-\frac{w^2}{2})~\text{d}w$ (e.g., Equation (67) in \cite{chandrasekaran2012convex}). Now by the definitions of $\psi(u,v)$ and $  \hat{\zeta}_{\rm PO}(\bx;\|\cdot\|_1)$, we conclude that $  \hat{\zeta}_{\rm PO}(\bx;\|\cdot\|_1)=n\psi(\frac{s}{n},\frac{\|\bx\|_1^2}{s})$, as desired.  \subsection{Proof of Theorem \ref{lowrankfor}}\label{provethm4}
  Suppose that the singular value decomposition of $\bX$ is given by  
\begin{align}
    \bX = [\bU_1~\bU_2]\begin{bmatrix}
    \bSigma_r & 0 \\
    0 & 0 
\end{bmatrix}\begin{bmatrix}
\bV_1^\top \\
\bV_2^\top
\end{bmatrix}, \label{svd}
\end{align}
where $\bU=[\bU_1~\bU_2]$ with $\bU_1\in \mathbb{R}^{p\times r}$, $\bU_2\in \mathbb{R}^{p\times (p-r)}$ and $\bV= [\bV_1~\bV_2]$ with $\bV_1\in \mathbb{R}^{q\times r}$, $\bV_2\in\mathbb{R}^{q\times (q-r)}$ are orthogonal matrices, and $\bSigma_r = \diag(\sigma_1,...,\sigma_r)$ are the diagonal matrix containing the singular values $\sigma_1\ge \sigma_2\ge ...\ge \sigma_r>0$. We shall first prove \begin{align}
    \frac{\hat{\zeta}_{\rm PO}(\bX; \|\cdot\|_{nu})}{pq}\to \Psi\Big(\rho,\nu,\frac{\|\bX\|_{nu}^2}{r}\Big).\label{first}
\end{align} With the convention in (\ref{svd}), by letting $\calL_{\bX}:\bA \mapsto \bA - (1-\sqrt{2/\pi})\langle \bA,\bX\rangle\bX$ be the linear operator in $\mathbb{R}^{p\times q}$ corresponding to $\bQ_{\bx}^{-1}$, the   subdifferential   and its scaled version $\bQ_{\bx}^{-1}\partial f(\bX)$ can be formulated as 
    \begin{gather*}
        \partial{f}(\bx) = \partial{(\|\cdot\|_{nu})}(\bX)= \left\{\bU\begin{bmatrix}
            \bI_r  & 0 \\
            0  &  \bW
        \end{bmatrix}\bV^\top: \|\bW\|_{op}\le 1 \right\},\\
        \bQ_{\bx}^{-1}\partial f(\bx) = \calL_{\bX}\partial{(\|\cdot\|_{nu})}(\bX) = \left\{\bU\begin{bmatrix}
            \bI_r - (1-\sqrt{2/\pi})\|\bX\|_{nu}\bSigma_{r} & 0 \\
            0  & \bW
    \end{bmatrix}\bV^\top:\|\bW\|_{op}\le 1\right\}.
    \end{gather*}
    Suppose $\bG\in \mathbb{R}^{p\times q}$ has i.i.d. $\calN(0,1)$ entries,  for any $\tau\ge 0$ we have
    \begin{align*}
       &\dist^2(\bG, \tau \calL_{\bX}\partial f(\bX)) \\&= \inf_{ \|\bW\|_{op}\le 1} \left\|\bG -\bU\begin{bmatrix}
            \tau\bI_r - \tau(1-\sqrt{2/\pi})\|\bX\|_{nu}\bSigma_{r} & 0 \\
            0  & \tau\bW
    \end{bmatrix}\bV^\top \right\|_F^2 \\
    & = \inf_{\|\bW\|_{op}\le 1}\left\|\bU^\top \bG\bV - \begin{bmatrix}
            \tau\bI_r - \tau(1-\sqrt{2/\pi})\|\bX\|_{nu}\bSigma_{r} & 0 \\
            0  & \tau\bW
    \end{bmatrix}\right\|_F^2\\
    & = \big\|\bG_{11}-\tau\bI_r +\tau(1-\sqrt{2/\pi})\|\bX\|_{nu}\bSigma_r\big\|_F^2 + \|\bG_{12}\|_F^2 + \|\bG_{21}\|_F^2 + \inf_{\|\bW\|_{op}\le 1}\|\bG_{22}-\tau\bW\|_F^2,
    \end{align*}
    where in the last equality we partition $\bU^\top\bG\bV \sim\calN^{p\times q}(0,1)$ into $$\begin{bmatrix}
        \bG_{11} & \bG_{12} \\
        \bG_{21} & \bG_{22}
    \end{bmatrix}~~\text{with $r\times r$ matrix $\bG_{11}$ and $(p-r)\times (q-r)$ matrix $\bG_{22}$.}$$  By taking expectation we arrive at
    \begin{align*}
       & \mathbbm{E}\dist^2(\bG, \tau \calL_{\bX}\partial f(\bX))  = \mathbbm{E}\Big[\|\bG_{11}\|_F^2+\tau^2r+ \tau^2(1-\sqrt{2/\pi})^2\|\bX\|_{nu}^2\\
        & \qquad-2(1-\sqrt{2/\pi})\tau^2\|\bX\|_{nu}^2 + \|\bG_{12}\|_F^2+\|\bG_{21}\|_F^2+  \inf_{\|\bW\|_{op}\le 1}\|\bG_{22}-\tau\bW\|_F^2 \Big]\\
        & = r\Big[p+q-r+\Big(1-\big(1-\frac{2}{\pi}\big)\frac{\|\bX\|_{nu}^2}{r}\Big)\tau^2\Big] + \mathbbm{E}  \Big[\inf_{\|\bW\|_{op}\le 1}\|\bG_{22}-\tau\bW\|_F^2 \Big] .  
        \end{align*}
    By substituting $\tau$ by $\sqrt{q-r}\cdot\tau$ and letting $\frac{r}{p}:=\bar{\rho}\in(0,1)$, $\frac{p}{q}:=\bar{\nu}\in(0,1]$, we obtain
    \begin{align*}
        & \hat{\zeta}_{\rm PO}(\bX;\|\cdot\|_{nu}) = \inf_{\tau\ge 0}\dist^2(\bG,\tau \calL_{\bX}\partial f(\bX)) = \inf_{\tau\ge 0}\dist^2\big(\bG,\sqrt{q-r}\cdot\tau \calL_{\bX}\partial f(\bX)\big)\\
        & = \inf_{\tau\ge 0}~ r\Big[p+q-r+\Big(1-\big(1-\frac{2}{\pi}\big)\frac{\|\bX\|_{nu}^2}{r}\Big)(q-r)\tau^2\Big] + \mathbbm{E}  \Big[\inf_{\|\bW\|_{op}\le 1}\|\bG_{22}-\sqrt{q-r}\cdot\tau\bW\|_F^2 \Big] \\
        \nn & = pq\cdot \inf_{\tau\ge 0}\Big\{ \frac{r}{p}\Big[\frac{p-r}{q}+1 + \Big(1-\big(1-\frac{2}{\pi}\big)\frac{\|\bX\|_{nu}^2}{r}\Big)\frac{(q-r)\tau^2}{q}\Big]+ \frac{(p-r)(q-r)}{pq}\mathbbm{E}  \Big[\inf_{\|\bW\|_{op}\le 1}\frac{\|\tilde{\bG}_{22}-\tau\bW\|_F^2}{p-r} \Big]\Big\}\\
        & = pq \cdot \inf_{\tau\ge 0}\Big\{\bar{\rho}\bar{\nu}+ \bar{\rho}(1-\bar{\rho}\bar{\nu})\Big[1+\tau^2\Big(1-\big(1-\frac{2}{\pi}\big)\frac{\|\bX\|_{nu}^2}{r}\Big)\Big]+(1-\bar{\rho})(1-\bar{\rho}\bar{\nu})\mathbbm{E}  \Big[\inf_{\|\bW\|_{op}\le 1}\frac{\|\tilde{\bG}_{22}-\tau\bW\|_F^2}{p-r}\Big]\Big\}\nn,
    \end{align*}
    where $\tilde{\bG}_{22} = {\bG_{22}}/{\sqrt{q-r}}$ has i.i.d. $\calN(0,(q-r)^{-1})$ entries. 
    Moreover, by the calculations in \cite[App. D.3]{amelunxen2014living} that is built upon Marcenko-Pastur law, as $p,q\to\infty$ with the limiting ratios $\bar{\rho}=\frac{r}{p}\to \rho \in(0,1)$ and $\bar{\nu}=\frac{p}{q}\to \nu\in(0,1]$, pointwise in $\tau\ge 0$ we have 
    \begin{align*}
        \mathbbm{E}  \Big[\inf_{\|\bW\|_{op}\le 1}\frac{\|\tilde{\bG}_{22}-\tau\bW\|_F^2}{p-r}\Big] \to  \int_{\max\{a_-,\tau\}}^{a_+}(b-\tau)^2\varphi_{y}(b)~\mathrm{d}b
    \end{align*}
    with $y,a_{\pm}$ and $\varphi_y(b)$ being defined in the statement. Moreover, since the involved functions are stricitly convex in $\tau$, we arrive at (\ref{first}); see \cite[App. D.3]{amelunxen2014living} for details. 
     We now invoke Theorem \ref{thm2} to obtain $$\hat{\zeta}_{\rm PO}(\bX;\|\cdot\|_{nu}) - O\big(\sqrt{r(p+q)}\big) - O\Big(\frac{p+q}{\|\bX\|_{nu}}\Big) \le\zeta_{\rm PO}(\bX;\|\cdot\|_{nu})\le \hat{\zeta}_{\rm PO}(\bX;\|\cdot\|_{nu}),$$which  implies \begin{align}
         \frac{\hat{\zeta}_{\rm PO}(\bX;\|\cdot\|_{nu})-\zeta_{\rm PO}(\bX;\|\cdot\|_{nu})}{pq}\to 0.\label{second}
     \end{align}  
     Combining (\ref{first}) and (\ref{second}) yields the desired claim. 
\end{appendix}
\end{document}